\documentclass[11pt]{article}

\usepackage{amsmath}
\usepackage{verbatim,color,amssymb}
\usepackage{amsmath}					
\usepackage{amsthm}					
\usepackage{natbib}
\usepackage{multirow}
\usepackage{setspace}
\usepackage[mathscr]{euscript}
\usepackage{fancyhdr}
\usepackage{enumitem}
\usepackage{graphicx}
\usepackage{subfigure}
\usepackage{lineno}
\usepackage{array,booktabs}
\usepackage{bm, bbold,bbm}
\usepackage{url}
\newcommand{\op}[1]{\Vert#1\Vert_{\mathrm{op}}}
\usepackage{subfigure}
\usepackage{lscape}

\usepackage{subfig}
\usepackage{setspace}
\usepackage{tikz}
\usetikzlibrary{arrows}
\usepackage{multirow}

\usepackage{subfigure}

\newtheorem{theorem}{{Theorem}}

\newtheorem{lem}{Lemma}
\newtheorem*{Proof*}{Proof}

\def\eE{\mathbb{E}}

\def\N{{\cal N}}

\usepackage{array}
\usepackage{verbatim}
\newcommand{\Frob}[1]{\Vert#1\Vert_{\mathrm{F}}}
\usepackage{amsmath}
\usepackage{amsfonts}
\usepackage{mathrsfs}
\usepackage{amssymb}
\usepackage{amsthm}
\usepackage{stmaryrd}

\usepackage{algorithm}
\usepackage{algorithmic}
\usepackage{stackengine}
\usepackage{graphicx}
\usepackage{multirow}
\usepackage{epstopdf}
\usepackage{multicol}
\usepackage{caption}
\usepackage{subfig}
\usepackage{enumitem}
\usepackage{amsmath}
\usepackage{pifont}
\usepackage{bbm}
\usepackage{makecell}
\usepackage{mathtools}
\usepackage{stfloats}
\usepackage[T1]{fontenc}
\usepackage{arydshln}
\usepackage{colortbl}
\usepackage{balance}
\usepackage{footmisc}

\usepackage{tikz}

\makeatletter
\newcommand*{\rom}[1]{\expandafter\@slowromancap\romannumeral #1@}
\makeatother

\def\diag{\hbox{diag}}
\def\Diag{\hbox{Diag}}

\def\diag{\hbox{diag}}

\def\sign{\hbox{sign}}
\def\trace{\hbox{tr}}

\def\vect{\hbox{vec}}

\def\Ind{\mathbbm{1}}

\def\Beta{\hbox{Beta}}

\def\IG{\hbox{IG}} %{\hbox{Inv-Ga}}

\def\Normal{\hbox{N}}

\def\Unif{\hbox{Unif}}

\def\bse{\begin{eqnarray*}}
\def\ese{\end{eqnarray*}}
\def\be{\begin{eqnarray}}
\def\ee{\end{eqnarray}}
\def\bq{\begin{equation}}
\def\eq{\end{equation}}

\def\trans{^{\rm T}}

\def\bone{{\mathbf 1}}

\def\b1e{\bm{e}}

\def\bA{\bm{A}}

\def\bB{\bm{B}}

\def\bC{\bm{C}}
\def\bd{\bm{d}}
\def\bD{\bm{D}}
\def\bE{\bm{E}}
\def\be{\bm{e}}

\def\bg{\bm{g}}

\def\bH{\bm{H}}
\def\bh{\bm{h}}
\def\bI{\bm{I}}
\def\bJ{\bm{J}}

\def\bM{\bm{M}}

\def\bP{\bm{P}}
\def\bq{\bm{q}}
\def\bQ{\bm{Q}}
\def\bR{\bm{R}}

\def\bS{\bm{S}}

\def\bU{\bm{U}}

\def\bV{\bm{V}}
\def\bw{\bm{w}}
\def\bW{\bm{W}}

\def\bX{\bm{X}}
\def\by{\bm{y}}
\def\bY{\bm{Y}}
\def\bfY{{\mathbf Y}}
\def\bz{\bm{z}}
\def\bZ{\bm{Z}}

\def\bS{\bm{S}}
\def\bzero{\mathbf{0}}

\newcommand{\etam}{\mbox{\boldmath $\eta$}}

\newcommand{\bDelta}{\mbox{\boldmath $\Delta$}}

\newcommand{\btheta}{\mbox{\boldmath $\theta$}}

\newcommand{\bbeta}{\mbox{\boldmath $\beta$}}

\newcommand{\bSigma}{\boldsymbol{\Sigma}}

\newcommand{\bOmega}{\mbox{\boldmath $\Omega$}}

\newcommand{\bLambda}{\mbox{\boldmath $\Lambda$}}

\newcommand{\bGamma}{\mbox{\boldmath $\Gamma$}}

\renewcommand\footnoterule{\kern-3pt \hrule \textwidth 2in \kern 2.6pt}

\def\boxit#1{\vbox{\hrule\hbox{\vrule\kern6pt \vbox{\kern6pt \textcolor{blue}{#1}\kern6pt}\kern6pt\vrule}\hrule}}

\def\authorfootnote#1{{\let\thefootnote\relax\footnotetext{#1}}}

\title{Bayesian Inference for High-dimensional Time Series with a Stationary Directed Acyclic Graphical Structure}
\author{Arkaprava Roy$^1$, Anindya Roy$^2$ and Subhashis Ghosal$^3$\\ {\it $^1$University of Florida, Gainesville, USA,}\\ $^2${\it University of Maryland Baltimore County, Baltimore, USA,} \\ {\it and $^3$North Carolina State University, Raleigh, USA}
}

\date{}

\begin{document}

\maketitle

\begin{abstract}
    In multivariate time series analysis, understanding the underlying marginal causal relationships among variables is often of interest for various applications. A directed acyclic graph (DAG) provides a powerful framework for representing causal dependencies. We propose a novel Bayesian approach for modeling multivariate stationary, possibly matrix-variate, time series where causal relations in the stationary graph structure are encoded by a DAG. The proposed model does not assume any pre-specified parent-child ordering.
    We use a novel ``immersion-posterior'' based efficient computational algorithm to draw Bayesian inference on quantities of interest. The posterior convergence properties of the proposed method are established along with two key identifiability results for the unrestricted structural equation models. We demonstrate superiority of the proposed method over several existing competing approaches via a limited simulation study.
    Finally, our analysis of the Quarterly Workforce Indices (QWI) time series based on employment volume and earnings reveals surprising stability of the causal structure within different industry-output-based salary groups for the middle-aged labor force, with changes occurring mainly at the extreme age groups. The DAG structure among the salary groups is found to be changing across the age groups. These findings corroborate the implicit dynamics in the labor market.
\end{abstract}


\section{Introduction}
Learning the structure of association, particularly causal association, between a large number of variables is one of the most common scientific quests in the modern data-driven world. Graphical models provide an appealing framework for characterizing associations by assessing the conditional dependence between variables in large-scale multivariate studies.  There is a vast literature on embedding a probabilistic model to study the graph structure ~\citep{lauritzen1996graphical, koller2009probabilistic} of a group of variables. A graph $\mathcal{G}$ is commonly denoted by the pair $\left(\bV,\bE\right)$, where $\bV$ denotes a set of nodes (variables) and $\bE$ denotes the set of edges (associations). An important subclass of graphical models are the Directed Acyclic Graph (DAG) models where the association between each pair of nodes is directed.  
One of the most popular probabilistic approaches to DAGs are the Gaussian-directed acyclic graphs (DAGs) \citep{babula2004dynamic,shojaie2010penalized} where the parent ordering of the nodes is assumed to be known. A source of their popularity is their for making causal discoveries \citep{pearl2000models} under suitable causal assumptions. Other models to study structural associations among the variables include the structural vector autoregressive class of models \citep{kilian2017structural,krampe2023structural,chan2025bayesian}.

However, the parent ordering is often unknown. When the ordering is unknown, the problem of learning the causal structure is much more challenging. In such situations, the DAG structure and the ordering both need to be learned from the data along with the strength of the dependence. 
Most works on DAG estimation rely on the structural equation model (SEM) which is also the premise of  our proposal. In  the  LINGAM approach \citep{shimizu2006linear}, an unrestricted estimate of the graphical association matrix is post-processed into a estimate under a DAG specification. Other approaches follow two-step procedures where the directional order is determined first, then the coefficients are estimated \citep{buhlmann2014cam, lee2022functional}.

Recently \cite{zheng2018dags} proposed an elegant continuous constrained optimization-based approach called NOTEAR using an algebraic characterization of DAG. Subsequently, several other continuous characterizations of the DAG constraint have been proposed \citep{yu2019dag,bello2022dagma}. These characterizations provide the opportunity for developing computationally efficient DAG estimation methods in a penalized framework using augmented Lagrangian optimization where additional penalties can be included to encourage sparsity and other structural properties.

Several approaches have been proposed for a fully Bayesian DAG estimation under known \citep{altomare2013objective, ni2015bayesian} and unknown parent ordering \citep{zhou2023functional}. In the case of unknown parent ordering, the MCMC-based computational algorithms are generally expensive. To alleviate that, two-stage procedures are often adopted where the parent ordering is identified in the first stage, and under that ordering, the DAG is estimated next \cite{shojaie2010penalized, altomare2013objective}.

Beyond the SEM framework, DAG can also be estimated using the PC-algorithm \citep{spirtes1991algorithm} that estimates the completed partially directed acyclic graph (CPDAG). \cite{kalisch2007estimating} studied the high dimensional uniform consistency of the PC algorithm. However, it relies on the faithfulness assumption
and is sensitive to individual failures of conditional independence tests. Another approach is the score-based \citep{heckerman1995learning} where an 
$M$-estimator for some score function is defined, and it is based on measuring how well the graph fits the data. Due to the large search space, the problem is, in general, NP-hard; thus, it resorts to local search.

When the nodes are time series, the problem of structure learning becomes considerably more challenging as one has to deal with dependence among the observations. Traditional approach in the time series literature involves using structural vector autoregressive models (SVAR) to
estimates the association structure between the lead and lags of the components of the multivariate time series; \cite{lutkepohl2005new, kilian2011structural, ahelegbey2016bayesian, Malinsky2018}. Recently \cite{pamfil2020dynotears} applied the NOTEAR approach on multivariate time series within a structural vector-autoregressive model (SVAR). 
\cite{runge2019detecting, runge2020discovering} propose estimation of causal relations in multivariate time series which looks at lead-lag causal associations along with the marginal ones.

For estimating causal associations among the component variables within the multivariate time series data, a somewhat different approach is based on the concept of Granger causality \citep{granger1969investigating}.  Granger causality is not based on the counterfactual arguments of a true causal framework but it does provide a sense of predictive relationships between the components of the time series.

The main emphasis of this article is on structure learning from  multivariate time series data as commonly found in applications such as economics \citep{imbens2020potential}, finance \citep{ji2018network}, neuroscience \citep{ramsey2010six}, and we aim to investigate and estimate the directional relationships among the time series. 
We focus on  `contemporaneous stationary causal' structure learning, where the causal structure is encoded in the marginal precision matrix of a stationary Gaussian time series. We characterize the contemporaneous precision matrix of a stationary time series to estimate the causal graph. In different applications, such estimates give rise to varying forms of inference and structure learning methodologies. 
Specifically, a multivariate time series can be represented as a DAG at each time point \citep{dahlhaus2003causality,ben2011high,zuo2012stock,deng2013learning}. Such a graphical dependence corresponds to a strictly lower-triangular adjacency matrix under an ordering of the nodes $V_1,\ldots,V_p$, such that $V_i$ cannot be a child of $V_j$ if $i < j$. If the process is stationary, the DAG structure stays invariant over time. Here, the nodes represent the variables. Such a causal form of dependence in a multivariate `stationary' time series $(\bY_t: t=0,1,\ldots)$ may be encapsulated through a linear transform, modeling the `causal residual process' $\bY_t-\bW \bY_t$ separately from the graph structure. The operator $\bW$ decorrelates the components in $\by_t$ at each time $t$ over the nodes. In order to capture DAG relationship, $\bW$ has to satisfy the DAG-ness constraint, that it can be permuted into a strictly lower triangular matrix by simultaneous equal row and column permutations.
Some approaches follow a two-step process: (1) determining the ordering using methods like order-MCMC \citep{friedman2003being, kuipers2017partition}  and (2) using that ordering to learn the network structure.
In a time-series setting, the first step itself is complicated.
We thus propose a novel projection-posterior-based approach.

We model $p$-dimensional residual $\bY_t-\bW \bY_t$ as a vector of independent univariate stationary time series. %in the same spirit as in the orthogonally rotated univariate time series (OUT) model of \cite{roy2024bayesian}.  Like the OUT model, 
The proposed model thus closely resembles the Dynamic Factor Model (DFM) and possesses a symmetric autocovariance matrix.
The model also allows  a `contemporaneous' DAG estimation, leading to the identification of marginal causal relationships among the variables.
The proposed characterization also leads to a stationary multivariate time-series model. The DAG-ness constraint on $\bW$ is imposed by a novel application of a ``projection-posterior''. The projected posterior method imposes the DAG-ness constraint on samples from the posterior
distribution of the $\bW$, applying the continuous constraint from \cite{zheng2018dags} along with an adaptive LASSO penalty \citep{zou2006adaptive} for each sample.  We also establish two identifiability conditions for the completely unrestricted SEM model to study posterior concentration properties. We show  identifiability when marginal error variances are known as well as identifiability without the known error variance assumption . In the second case we introduce a linear independence condition among the spectral densities of the components in the error process $(\bY_t-\bW \bY_t)$ to facilitate identification of the main parameters. The second result specifically leverages the time-series nature of the data. We also characterize the condition of linear independence of spectral densities of the unobserved latent process, thus allowing an empirical verification of the condition from the observed data. Unlike the existing identifiability results in the literature, our results show identifiability of the completely unrestricted SEM model.

The paper makes several key theoretical and methodological contributions to the literature of DAG models for time series: 
\begin{enumerate}
\item  It develops a novel `immersion-posterior' based Bayesian approach for estimating a `stationary' DAG structure in a semiparametric multivariate time series without assuming any parent-child ordering.
\item 
The paper establishes two key identifiability results for the completely unrestricted SEM model and establishes posterior concentration under these identifiability conditions. The second identifiability result explicitly exploits the time-series structure of the data. 
\item 
The proposed approach  is extendable to DAG modeling of matrix-variate time series, and its potential is further illustrated by an application. 
\item 
The analysis of the Quarterly Workforce Indicator series presented in the paper reveals  interesting DAG-based association patterns across various industries, highlighting differences in earning profiles across age groups.
\end{enumerate}

 The paper is organized as follows. In Section~\ref{sec:method}, we lay out our projection-posterior-based Bayesian model along some extensions motivated by the real data analysis. The pseudo-likelihood based on the Whittle approximation, specification of prior distributions, and posterior sampling strategies are described in Section~\ref{sec:prior}. The convergence properties of the posterior distribution are studied in Section~\ref{sec:convergence}. Extensive simulation studies to compare the performance of the proposed Bayesian method with other possible methods are carried out in Section~\ref{sec:simu}. A dataset on Quarterly Workforce Indicators is analyzed in Section~\ref{sec:realdata}.
 
\section{Incorporating DAG structure through projection posterior}
\label{sec:method}

We consider the structural equation model \citep{pearl1998graphs}, but allow the error to be an array of stationary processes. Specifically, we assume that the multivariate time series $\{\bY_t: t=1,\ldots,T\}$ is given by the DAG structure 
\begin{align}
    \bY_t-\bW \bY_t=&\bD^{1/2} \bZ_{t}, \label{eq:SEM}
\end{align}
where $\bW$ is a $p\times p$ matrix controlling the DAG structure, $\bD$ is a diagonal matrix with positive entries, and $\{\bZ_t: t=1,\ldots,T\}$ is a stationary process with independent components having unit variance. To follow a conventional  Bayesian approach, the prior on $\bW$ has to comply with the restriction that $\bW$ can be permuted into a strictly lower triangular matrix by simultaneous equal row and column permutations. Such a restriction is difficult to impose, and the corresponding posterior distribution would be extremely complicated, undermining posterior sampling, and convergence properties of the posterior distribution would be hard to study. To avoid the problem, we take the projection-posterior approach used earlier by \cite{lin2014bayesian,chakraborty2021convergence,chakraborty2021coverage,wang2012bayesian} for shape-restricted models and \cite{pal2024bayesian} for sparse linear regression. In this approach, the restriction is not imposed in the prior, but posterior samples are projected to comply with the desired restriction. The induced posterior distribution, called the projection posterior (or, more generally, the immersion posterior when the transformation is not necessarily a projection with respect to a distance function) is used to make an inference. As a result, a simple prior can be used, allowing efficient posterior sampling. By a simple concentration inequality (Equation (3.2) of \cite{chakraborty2021convergence}), the projection posterior inherits the contraction rate of the unrestricted posterior. As the latter can often be established from the general theory of posterior contraction \citep{ghosal2017fundamentals}, the contraction rate of the projection posterior can be obtained. 

In the present model, we only restrict $\bW$ to have zeroes in the diagonal.
To complete our model characterization, the distribution of the $i$th component process in $\bZ_{t}$ is modeled in the frequency domain by a spectral density $f_i$. We then impose priors on these parameters, with details provided in Sections~\ref{sec:prior}. Section~\ref{sec:sampling} provides the steps to get the `unrestricted' posterior samples of our model parameters $(\bW,\bD,f_1,\ldots,f_p)$. The samples of $\bW$ obtained here do not satisfy the DAG-ness constraint but are sparse due to the horseshoe prior and have zeroes in the diagonal as specified by the model construction. These samples are then passed through an immersion map that transforms these sparse unrestricted matrices with zero diagonals into sparse matrices satisfying the DAG constraint.

We set an immersion map as a solution to the minimization problem 
\begin{align} 
\label{projection operator}
\vartheta:\bW\rightarrow\bW^{*}:=\arg\min_{\bbeta} \left\{ \frac{1}{Tp} \|\bW \bY-\bbeta \bY\|_2^2 +\lambda\|\bC\odot\bbeta\|_1+\alpha h(\bbeta) + \frac{\rho}{2} h^2(\bbeta): {\bbeta\in \mathbb{R}_{*}^{p\times p}}\right\}; 
\end{align} 
where,  $h(\bbeta)=\trace(\exp\{\bbeta\odot\bbeta\})-m$ and $\mathbb{R}_{*}^{p\times p}$ stands for the set of all $p\times p$ matrices with `zero' diagonal, and tuning parameters $\lambda$ and $\bC$, standing for the LASSO penalty and the adaptive LASSO weights, respectively. The last two terms in the expression help impose the DAG-ness restriction (c.f.,  \cite{zheng2018dags}), and the adaptive LASSO \citep{zou2006adaptive} penalty 
in the second term forces sparsity in the DAG structure as in \cite{xu2022sparse}. While more choices proposed in \cite{yu2019dag,bello2022dagma} may also be applied to impose DAG-ness, we found that the penalty induced by the choice $h(\bbeta)=\trace(e^{\bbeta\odot\bbeta})-m$ performs the best in our time-series setting. 
The tuning parameters $\lambda$ and $\bC$ of the adaptive LASSO map are obtained once using the full data and kept fixed throughout the posterior sampling stage. Specifically, we let $\bC=1/|\hat{\bbeta}_{A}|^{\zeta}$ entrywise, where 
\begin{align} 
\label{alasso}
\hat{\bbeta}_{A}=\arg\min \{ \frac{1}{Tp} \|\bY-\bbeta \bY\|_2^2 +\alpha h(\bbeta) + \frac{\rho}{2} h^2(\bbeta): {\bbeta\in \mathbb{R}_{*}^{p\times p}}\}, 
\end{align}
and $\zeta$ is set to $1$ as default. Although $\zeta$ may be selected by cross-validation, the simpler choice $\zeta=1$ works well in our numerical experiments. To tune $\lambda$, we reparametrize $\bbeta$ to $\bm{\eta}=\bC\odot\bbeta$ and apply the standard LASSO cross-validation procedure to the optimization problem $(Tp)^{-1}\|\bY-(\bm{\eta}\oslash\bC) \bY\|_2^2 +\lambda\|\bm{\eta}\|_1$
using the R function {\tt cv.glmnet} from R package {\tt glmnet}, where the operator $\oslash$ stands for entrywise matrix division. 

To compute the immersion map \eqref{projection operator}, we solve the minimization problem using the R package {\tt lbfgs}. Since the LASSO penalty is not differentiable, following \cite{zheng2018dags}, we rewrite the loss and set the estimator $\hat{\bbeta}$ as $\hat{\bbeta}_1-\hat{\bbeta}_2$, where $(\hat{\bbeta}_1,\hat{\bbeta}_2)$ is obtained as the minimizer of  
\begin{align*} 
\frac{1}{Tp} \|\bW_{U}^{(t)} \bY-(\bbeta_1-\bbeta_2) \bY\|_2^2 +\lambda\sum_{i,j}c_{i,j}\beta_{1,i,j}+\lambda\sum_{i,j}c_{i,j}\beta_{2,i,j}+\alpha h(\bbeta_1-\bbeta_2) + \frac{\rho}{2} h^2(\bbeta_1-\bbeta_2) 
\end{align*} 
subject to the restriction that $\beta_{\ell,i,j}\geq 0$ for all components and $\bW_{U}^{(t)}$ is $t$-$th$ `unrestricted' posterior sample of $\bW$. 
The following algorithm can describe the posterior sampling steps. 

\begin{algorithm}[H]
%\SetAlgoLined
Step 1: Set adaptive LASSO weights $\bC=1/|\hat{\bbeta}_{A}|$,  where $\hat{\bbeta}_{A}$ is given by \eqref{alasso}. 

Step 2: Select $\lambda$ by minimizing $\min \{ (Tp)^{-1} \|\bY-(\etam\oslash\bC) \bY\|_2^2 +\lambda\|\etam\|_1:\etam\in \mathbb{R}_{*}^{p\times p}\}$ using cross-validation.\\

Step 3: Compute projected samples $\bW^{(t)}=\vartheta(\bW_U^{(t)})$ given by 
\eqref{projection operator}, where $\bW_U^{(t)}$ is the $t$th `unrestricted' posterior sample of $\bW$ i.e. $\vartheta(\bW_U^{(t)})=\bW^{(t)}$. 

Step 4: Set a small threshold $H$ and compute the DAG-adjacency matrix $\hat{\bA}$ with the $(i,j)$th entry given by  $a_{i,j}=\Ind\{B^{-1}\sum_{t=1}^B\bone\{|w^{(t)}_{i,j}|>H\}>0.5\}$.
\caption{Projection posterior samples for DAG}
\label{algo1}
\end{algorithm}

In Step 2, we use the function {\tt cv.glmnet} from the R package {\tt glmnet}.
The thresholding value in Step 4 is set at $H=0.3$ as the default in all of our numerical experiments.
In our data application, too, this choice was reasonable as the estimated DAG was found to be stable around this cutoff while examining sensitivity.

\subsection{Extension to matrix-variate time-series}
\label{sec:extensionmodel}

We now extend our model in~\eqref{eq:SEM} to $p\times S$-dimensional matrix-variate data, denoted by $\bfY_t=(\bY^{(s)}_{t}: s=1,\ldots,S)$, observed time points $t=1,\ldots,T$, where  $\bY^{(s)}_{t}$ is a $p$-dimensional time series. 
This extension is motivated by the QWI application in Section~\ref{sec:realdata}. In this setting, we are only interested in a causal association among the row variables of $\bY_t$; %Section~\ref{sec:twodag} shows how to incorporate the DAG structure in both directions. 
Specifically, we are interested in a common DAG characterizing the causal associations within $\bY^{(s)}_{t}$. To achieve this, we consider a two-layer model, combining the model from the previous section and the OUT model \citep{roy2024bayesian} in layer two on the residuals. % following an approach similar to the seemingly unrelated regression.
\begin{align*}
\bY^{(s)}_{t}-\bW \bY^{(s)}_{t}=& \bD^{1/2}\bX^{(s)}_{t},\quad \bX_{t}^{(s)}=(X^{(s)}_{k,t}: k=1,\ldots,p),\quad s=1,\ldots,S,\\
    \bX_{k,t}=&\bB\bU\bZ_{k,t},\quad k=1,\ldots,p,\quad \bX_{k,t}=(X^{(s)}_{k,t}: s=1,\ldots,S),
\end{align*}
where $\bZ_{k,t}=(Z_{s,j,t}: s=1,\ldots,S)$ are independent univariate stationary time-series with unit variance having spectral densities $f_{s,j}(\omega)$, $\bU$ is an orthogonal matrix and $\bB$ is the spherical coordinate representation of Cholesky factorizations of the correlation matrix across the groups. Thus, $\bQ=\bB\bB\trans$ is a correlation matrix.

The above model marginally at each time-point assumes a matrix-normal distribution on $\bR_{t}=(\bR^{(1)}_{t}, \ldots,\bR^{(S)}_{t})$, where $\bR^{(s)}_{t}=\bY^{(s)}_{t}-\bW \bY^{(s)}_{t}$.
Specifically, $\bR_{t}$ marginally follows the matrix-normal distribution with parameters $(\bzero, \bD,\bB\bB\trans)$; see \cite{gupta2018matrix} for the definition.

\subsubsection{Polar representation of $\bB$}

Following \cite{zhang2011new}, we represent the $S\times S$ lower-triangular matrix $\bB=(\!(b_{jk}: j,k=1,\ldots,S)\!)$ such that $\bQ=\bB\bB\trans$ is a correlation matrix through the relations 
\begin{align*}
b_{1,1} =1, \quad 
b_{2,1}=\cos a_{2,1}, \; b_{2,2}=\sin a_{2,1},\quad 
b_{3,1} =\cos a_{3,1}, \; b_{3,2}=\sin a_{3,1}\cos a_{3,2}, \; b_{3,3}=\sin a_{3,1},
\end{align*} 
and for $\ell=4,\dots,S$ and $k=2,\dots,\ell-1$, $b_{\ell,1}=\cos( a_{\ell,1})$,
\begin{align*} 
%&& v_{\ell,2}=b_{\ell-2}\cos a_{i_{1}(\ell)}\sin a_{i_{1}(\ell)+1},\dots,\\
%&& v_{\ell,\ell-2}=b_{\ell-2}\cos a_{i_{1}(\ell)}\cos a_{i_{1}(\ell)+1}\dots \cos a_{i_{2}(\ell)-1}\sin a_{i_{2}(\ell)},\\ 
b_{\ell,k}=\sin a_{\ell,1}\sin a_{\ell,2}\cdots \sin a_{\ell,k-1}\cos a_{\ell,k-1}, \quad 
 b_{\ell,\ell}=\sin a_{\ell,1}\sin a_{\ell,2}\cdots \sin a_{\ell,k-1}\sin a_{\ell,\ell-1},
\end{align*} 
and $0\leq  a_{j,i}\leq \pi$ for $1\leq i\le j\leq S-1$, making all diagonal entries positive. A sparsity-inducing prior is applied to $a_{i,j}$'s and illustrated in Section~\ref{sec:prior}.

\subsection{Modeling univariate error processes}

Since our final data application relies on the model in Section~\ref{sec:extensionmodel}, we discuss the model for the error processes for the multi-group case directly. As in \cite{roy2024bayesian}, the latent univariate processes $(Z_{s,k,t}:t=1,2,\ldots)$, $s=1,\ldots,S,\;k=1,\ldots,p$, are assumed to be independent and stationary.
We again model these univariate time series in the spectral domain and model the spectral density of $(Z_{s,k,t}:t=1,2,\ldots)$, defined by 
$$f_{s,k}(\omega)=\gamma_j(0)+2\sum_{h=1}^\infty \gamma_{s,k}(h) \cos (h\omega),\quad  \omega\in [-\pi,\pi].$$ 
Furthermore, each spectral density is symmetric about $0$ and uniquely determines the distribution of the time series since the autocovariances $(\!(\gamma_{s,k}(h))\!)$ are given by the inverse Fourier coefficients 
$$\gamma_{s,k}(h)=\int_{-\pi}^\pi f_{\ell}(\omega) \cos(h \omega)d\omega, \quad  h=1,2,\ldots.$$ 
To ensure unit marginal variances of each latent component time series $Z_{s,k,t}$, we need to impose the restriction 
\begin{align}
 \int_{-\pi}^\pi f_{s,k}(\omega) \cos(h \omega)d\omega=\gamma_{\ell}(0)=1, \;  s=1,\ldots,S,\;k=1,\ldots,p.
 \end{align}
Hence, the spectral densities $f_{s,k}$, $s=1,\ldots,S,\;k=1,\ldots,p$, are symmetric probability densities on $[-\pi,\pi]$.

\section{Likelihood, prior distribution and posterior sampling}
\label{sec:prior}

\subsection{Likelihood}
The likelihood is based on the distributions of the univariate time series $Z_{k,t}$'s. More precisely, we use the Whittle likelihood based on the approximately independent distributions of the Fourier transforms of the data at the Fourier frequencies $\omega_1,\ldots,\omega_{\lfloor T/2\rfloor}$, which are centered normal with variances given by  
\begin{eqnarray}
\bS_1 &=& \Diag(f_1(\omega_1), \ldots, f_p(\omega_1)),\nonumber \\
\bS_{2k+1} = \bS_{2k} &=& \Diag(f_1(\omega_k), \ldots, f_p(\omega_k)), \;\; k = 1, \ldots, \lfloor(T-1)/2\rfloor, \nonumber \\
\bS_T &=& \Diag(f_1(\omega_{T/2}), \ldots, f_p(\omega_{T/2})), \; \mbox{for even } T,
\label{eq:st}
\end{eqnarray}
where for any vector $\bm{a}$, $\Diag(\bm{a})$ refers to the diagonal matrix with entries $\bm{a}$ in that order. Hence, the log-likelihood based on Whittle's approximation is given by 
\begin{align}
\log q_T(\bY)=  -\frac{pT}{2}\log (2\pi)+ \frac12\sum_{t=1}^T \log\det \bM_t - \frac{1}{2} \sum_{t=1}^{T}  \bY_t\trans  \bM_t \bY_t,
   \label{eq:log-likelihood}
\end{align}
with 
$\bM_t=(\bI-\bW) \bD^{-1/2} \bS_t^{-1} \bD^{-1/2} (\bI-\bW)^T$. 

\subsection{Prior distribution}

We put the commonly adopted nonparametric Bayesian prior given by a finite random series for the spectral densities. We choose the standard B-spline basis to form the series because of their shape and order-preserving properties.  Writing $B_{j}^*=B_j/\int_0^1 B_j(u)du$ for the normalized B-splines with a basis consisting of $J$ many B-splines, we may consider a model indexed by the vector of spline coefficients $\btheta_{s,k}=(\theta_{s,k,1},\ldots,\theta_{s,k,j})$ given by 
\begin{align}
    \varphi(\omega;\btheta_{s,k})= \sum_{j=1}^J \theta_{s,k,j}B_j^*(|\omega|/\pi),
    \quad 
    \theta_{s,k,j}\ge 0, \quad \sum_{j=1}^J \theta_{s,k,j}=1/2.    
    \label{eq:specden}
\end{align}

We consider the following convenient representation that avoids the restriction of nonnegativity and sum constraint:
\begin{align} 
\theta_{s,k,j}={\Psi(\kappa_{s,k,j})}/\{2\sum_{j=1}^J\Psi(\kappa_{s,k,j})\},  
\label{eq:spline in kappa}
\end{align} 
where $\Psi(u)=(1+u/(1+|u|))/2$ is a link function monotonically mapping the real line to the unit interval \citep{roy2024bayesian}. The mild polynomial tail of $\Psi$ helps better control distance on the $\theta$-parameters in terms of the $\kappa$-parameters, essential for deriving the posterior contraction rate. 

A further dimension reduction can be achieved by assuming a CANDECOMP/PARAFAC (CP) low-rank decomposition  
\begin{align}
\kappa_{s,k,j}=\sum_{r=1}^R\xi_{s,r}\chi_{k,r}\eta_{j,r}, \quad s=1,\ldots,S, \; k=1,\ldots,p, \; j=1,\ldots,J, 
\label{eq:kappa low rank tensor}
\end{align}
reducing the dimension from $pSJ$ to $R(p+S+J)$, boosting the posterior contraction rate.

We now illustrate the prior distributions for each parameter involved in our model --- the entries of $\bW$, $\bD$, and the coefficients appearing in the B-spline basis expansion of the spectral densities.

\begin{itemize}
    \item Horseshoe prior on $\bW$: For $i\neq j$, we let $w_{i,j}\sim\Normal(0,d_{i,i}\lambda^2_{i,j}\tau^2)$ with  $\lambda_{i,j},\tau\sim \mathrm{C}^{+}(0,1)$, where $\mathrm{C}^{+}$ stands for the half-Cauchy distribution. The half-Cauchy prior from \cite{carvalho2010horseshoe} admits a scale mixture representation \citep{makalic2015simple}:  
    \begin{align}
    \lambda^2_{i,j}\sim\text{IG}(1/2,1/\nu_{i,j}), \quad \tau^2\sim\text{IG}(1/2,1/\xi), \qquad \nu_{i,j},\xi\sim\text{IG}(1/2,1), 
    \end{align}
    where IG stands for the inverse-gamma distribution. 
    This representation allows a straightforward Gibbs sampling procedure. Marginally $\lambda_{i,j},\tau\sim \mathrm{C}^{+}(0,1)$.
    
    \item Soft-thresholding prior on the parameters $a_{i,j}$ of the polar representation: We put a soft-thresholding prior to promote shrinkage on $a_{i,j}$ at $\pi/2$.
$a_{i,j}=\sign(a^*_{i,j}-\pi/2)(|a^*_{i,j}-\pi/2|-\lambda)_{+}+\pi/2$ and ${a^*_{i,j}}/{\pi}\sim\Beta(a,a)$ or $a^*_{i,j}=\pi(1+a^u_{i,j}/(1+|a^u_{i,j}|))/2$ with $a^u_{i,j}\sim\Normal(0,\sigma_T^2)$. Finally, we set a uniform prior for the soft-thresholing parameter $\lambda\sim \mathrm{Un}[\lambda_L,\lambda_U]$.
   
    \item Stick-breaking mixture prior for $\bD$: To induce clustering among the entries of $\bD$, we let the entries of $\bD$ be independent and identically distributed (i.i.d.) following a distribution $G$ given a stick-breaking prior through the representation $G=\sum_{i=k}^M \pi_k\delta (m_k)$ with $m_{k}$ following an Inverse-Gaussian distribution \citep{chhikara1988inverse} with density function proportional to $t^{-3/2} e^{-(t-\mu_d)^2/(2t)}$, $t > 0$, for some $\mu_d>0$ and the random weights are given by the strick-breaking process $p_k=v_k\prod_{j=1}^{k-1}v_{j}$ with $v_j\sim$Beta$(1,v)$. Further, we put a weakly informative centered normal prior on $\mu_d$ with a large variance and let $v\sim\mathrm{Ga}(a_v,b_v)$, the gamma distribution with shape parameter $a_v$ and rate parameter $b_v$.
    
    \item The parameters in the CP-decomposition: Since an appropriate value of the rank $R$ is unknown, we consider an indirect automatic selection of $R$ through cumulative shrinkage priors \citep{bhattacharya2011sparse}. % on $\xi_{s,r}$ and $\chi_{k,r}$ $s=1,\ldots,S, \; k=1,\ldots,p$, $r=1,\ldots,R$. 
    \begin{itemize} 
   \item  
   For all $r=1,\ldots,R$, we put independent priors 
$$\xi_{s,r}|v_{s,r},\tau_{r}\sim \mathrm{N}(0,v_{s,r}^{-1}\tau_{r}^{-1}), \quad v_{s,r}\sim \mathrm{Ga}(\nu_1,\nu_1),\quad \tau_{r}=\prod_{i=1}^r\Delta_{i},$$
%\xi_{j,r}\sim \mathrm{N}(0,\sigma_\xi^2), \quad  
where $\Delta_{1}\sim \mathrm{Ga}(\kappa_{1}, 1)$ and $\Delta_{i}\sim \mathrm{Ga}(\kappa_{2}, 1)$, $i\geq 2$. The parameters $v_{j,r}$, $j=1,\ldots,p$, $r=1,2,\ldots$, control local shrinkage of the elements in $\xi_{j,r}$, whereas $\tau_r$ controls column shrinkage of the $r$th column.
\item
 Next, let $\etam_r=(\eta_{1,r},\ldots,\eta_{J,r})\sim\mathrm{N}_K (0,\sigma_{\kappa}\bP^{-1}), \quad  
\sigma_\xi,\sigma_{\kappa}\sim\IG(c_1,c_1)$; here 
$\bP$ is the second-order difference matrix to impose smoothness given by  $\bP=\bQ\trans \bQ$, where $\bQ$ is the $K\times(K+2)$ matrix of the second difference operator. 
\item 
On $J$, we put a prior with a Poisson-like tail $e^{-J \log J}$. 
\item 
Finally, another cumulative shrinkage prior is set for $\chi_{k,r}$.
\end{itemize}
\end{itemize}

\subsection{Posterior sampling}
\label{sec:sampling}

To obtain a sample from the unrestricted posterior of all the parameters, we execute the following steps: 
\begin{itemize}
    \item Updating $\bW$: Each row in $\bW$ enjoys full conditional Gaussian posterior.
    \item Updating $\bD$: We introduce latent indicator variables $z_{\ell}$'s for each diagonal entry $d_{\ell,\ell}$ in $\bD$. The posterior of the atoms $m_k$'s are Generalized Inverse-Gaussian and are updated
    \item Updating $\bB$: We consider adaptive Metropolis-Hastings \citep{haario2001adaptive} to update the latent polar angles $a^*_{i,j}$'s and the thresholding parameter $\lambda$ is updated based on a random-walk Metropolis-Hastings in log-scale with a Jacobian adjustment.
    \item Updating the spectral densities: We consider gradient-based Langevin Monte Carlo (LMC) to update the parameters involving $\kappa_{s,k,j}$'s as in \cite{roy2024bayesian}.
\end{itemize}
After obtaining the sample for $\bW$, we transform it to $\bW^*$ via the projection map \eqref{projection operator} and record the value of $\bW^*$, while the next iteration of the MCMC follows with the untransformed value.

\section{Posterior contraction}
\label{sec:convergence}

In this section, we show that the projection posterior distribution introduced in the two preceding sections converges to an assumed true parameter value with a DAG structure and compute the corresponding posterior contraction rate. We simplify the notations in the analysis, restrict to the simpler setting of no multi-task components in the matrix-variate case, and consider the simpler penalty without adaptation by scaling by a preliminary estimate. 
The result for the matrix-variate case can be established by extending the results presented here with additional results from \cite{roy2024bayesian}.
We make the following assumptions on the true values $(\bD_0, \bW_0, f_{0,1},\ldots,f_{0,p})$ of the parameters $(\bD, \bW, f_{1},\ldots,f_{p})$. 

\vskip10pt
\underline{Conditions on true distribution:}
\begin{description}
\item [(A1)] The true values $(\bD_0, \bW_0)$ of $(\bD,\bW)$ are such that 
\begin{itemize} 
\item the diagonal entries of $\bD_0$ are bounded above and below by two positive constants; 
\item $\|\bW_0\|_{\infty}$ is bounded and $\bW_0$ has $s$ non-zero entries; 
\item without loss of generality, $\bW_0$ is strictly lower triangular. 
\item The true precision matrix $\bOmega_0$ has eigenvalues that are bounded and bounded away from zero.
\end{itemize} 
\item [(A2)] The true spectral densities $f_{0,1},\ldots,f_{0,p}$ of the latent processes are positive and H\"older continuous with a smoothness index $\alpha>0$ (cf.,  Definition~C.4 of \cite{ghosal2017fundamentals}) such that 
\begin{align}
\label{eq:spline approximation}
    \max_{1\le j\le p} \{ \sup \{|f_{0,j}(\omega)-\sum_{k=1}^K \theta_{jk}^* B_k^*(\omega)|: \omega \in [0,1]\}\lesssim K^{-\alpha},
\end{align}
for some uniformly bounded sequence $\theta_{jk}^*$ of the form \eqref{eq:spline in kappa} and \eqref{eq:kappa low rank tensor}.  
 
\item [(A3)] 
Growth of dimension: $\log p\lesssim \log T$. 
\end{description}

It is well-known that H\"older functions with smoothness index $\alpha$ can be approximated within $O(K^{-\alpha})$ by a B-spline basis expansion with $K$ terms. Condition (A2) incorporates dimension reduction in the spectral densities by assuming that the approximation for the true spectral densities is not affected by restricting the spline coefficients by the relation \eqref{eq:kappa low rank tensor}.

Due to the complexity in nonparametric time-series modeling, we need to impose the following two restrictions on the support of $\bD$, $\bW$, and the spectral density parameters. Similar assumptions were used in Theorem 3 of \cite{roy2024bayesian}.

\underline{Conditions on the prior:}
\begin{enumerate}
    \item [(P1)] Diagonal entries $d_1,\ldots,d_p$ and the range of the functions $f_1,\ldots,f_p$ lie in a fixed, compact subinterval of $(0,\infty)$.
    \item [(P2)] There exist constants $0<c_H\leq C_H<\infty$, such that the prior of $\bW$ is supported on
\(
    \mathcal W_{c_H,C_H}
    =
    \left\{
        \bW:\ \operatorname{diag}(\bW)=0,\quad
        \|\bI-\bW\|_{\mathrm{op}}\leq C_H,\quad
        \|(\bI-\bW)^{-1}\|_{\mathrm{op}}\leq c_H^{-1}
    \right\}.
\)
The true matrix $W_0$ belongs to the interior of
$\mathcal W_{c_H,C_H}$.
\end{enumerate}
Condition (P2) is needed for the unrestricted structural transformation
$\bI-\bW$ to remain uniformly well conditioned. We assume that the true DAG
matrix $\bW_0$ satisfies this condition. For $\bW_0$, after a suitable
permutation,
\(
    (\bI-\bW_0)^{-1}=\bI+\bW_0+\cdots+\bW_0^{p-1},
\)
so the bound on $\|(\bI-\bW_0)^{-1}\|_{\mathrm{op}}$ controls the cumulative
effects transmitted along directed paths in the true DAG. 
One can set the lower and upper bounds $c_H$ and $C_H$, respectively, based on such a reasonable understanding on the true DAG.
The
theoretical prior on the unrestricted zero-diagonal matrix $\bW$ is then
restricted to a well-conditioned parameter space containing $\bW_0$ in
its interior. This condition is similar in spirit to the graphical LASSO-type priors \citep{banerjee2014posterior}.  We can incorporate it within our sampling mechanism.
We first define $\bOmega=(\bI-\bW)\trans\bD^{-1}(\bI-\bW)$ and similarly $\bSigma=(\bI-\bW)^{-1}\bD\{(\bI-\bW)\trans\}^{-1}$. Let 
\begin{equation}
	\label{equ6}
	%\small
	\begin{aligned}
	L_{T}(\bW,\bbeta)=\frac{1}{T}\|\bW \bY-\bbeta \bY\|_2^2+\zeta_T\sum_{i,j}c_{i,j}|\beta_{ij}|+\frac{\rho}{2}|h(\bbeta)|^2+\alpha h(\bbeta),
	\end{aligned}
\end{equation}
where to compute the posterior, $\bW$ is not restricted to a DAG-structure but only has to satisfy the zero-diagonal conditions. And the projection map $\nu:\bW\rightarrow\bW^{*}:=\arg\min_{\bbeta} L_{T}(\bW,\bbeta).$%We also assume $\zeta_T/T\rightarrow \zeta_0$. 
%Let $\eta_T$, 

\begin{theorem}
    \label{contraction rate}
    Under Conditions {\rm (A1)--(A3)} on the true parameter and Conditions {\rm (P1)--(P2)} on the prior, we have that in true probability 
    \begin{itemize}
        \item [{\rm (i)}]  if $\bD$ is known, $\Pi(\|\bW^*-\bW_0\|>m_T \sqrt{p}\sqrt{\epsilon_T}|\bY)\to 0$  for all $m_T\to \infty$, where 
        \begin{align}
        \epsilon_T=\max(\sqrt{(p+s)/T}, T^{-\alpha/(2\alpha+1)}) \sqrt{\log T};
        \label{eq:contraction rate}
        \end{align}
        \item [{\rm (ii)}]  if $\bD$ is unknown, and the spectral densities $f_1,\ldots,f_p$ are linearly independent functions, then projection-posterior distribution for $\bW$ concentrates at $\bW_0$, that is, for any $\epsilon>0$, $\Pi(\|\bW^*-\bW_0\|> \epsilon|\bY)\to 0$. 
    \end{itemize}
    \label{eq:Theorem1}
\end{theorem}

The proof of the theorem is given in the supplementary material. The main objective of the proposed methodology is to recover the DAG structure given by the matrix $\bW$. To do so, we show that the posterior concentrates around the true value.  In the proof of Theorem~\ref{contraction rate}, we first establish  concentration properties of the unrestricted posterior, and then show that the posterior induced by the map $\bW\mapsto \bW^*$  inherits those concentration properties. While $\bW\mapsto \bW^*$ is non-convex, and hence not a projection but an immersion, the concentration properties still transport from the unrestricted posterior to the projected posterior due to the local properties of the map in the neighborhood of the true value (see proof for details). 

Without the boundedness conditions, concentration of the unrestricted posterior can still be established in terms of the average truncated squared Frobenius distance on the Whittle transform of each component series. This is weaker than the assertions in Theorem~1, and the result would be analogous to that given in  Theorem~2 of \cite{roy2024bayesian}. To satisfy the bounded range requirements in Theorem~1, the priors on $\bD$ and $\chi_{k,r}$ and $\eta_{j,r}$, $k=1,\ldots,p$, $j=1,\ldots,J$ in the representation \eqref{eq:kappa low rank tensor} can be truncated. These conditions are similar to the restriction on the error variances used in \cite{loh2014high} for a high-dimensional setting without time dependence.

\section{Simulation}
\label{sec:simu}
We run three simulation experiments to evaluate our proposed model's performance. 
We generate the simulated datasets following the proposed model with $p=40, S=15$, and $T=32$ or $48$ with three different choices for the distributions of the univariate time series. The specific differences are described at the beginning of each subsection. For all the cases, the sparse DAG matrix, $\bW$, is generated using {\tt randDAG} of {\tt pcalg}  \citep{pcalgR} with two possible values for the `Expected neighbors' as 2 and 4. This leads to different levels of sparsity in $\bW$. The weights for the edges are generated from $\Unif((-2,-0.5)\cup(0.5,2))$.

The entries in $\bD$ are generated as absolute values of $\Normal(7,2)$.
The $S\times S$ correlation matrix $\bB$ is generated in two steps. First, a precision matrix is generated using g-Wishart where the underlying graph is simulated by combining three small-worlds, each with 5 nodes similar to \cite{roy2024bayesian}.
Then, we take its inverse and scale to get the correlation matrix $\bB$.

 We compare our methods with NO-TEARS \citep{zheng2018dags} in terms of both estimation of $\bW$ and identification of edges. We also compared with PC and LINGAM, but only in identifying DAG edges. There is a polynomial version of \cite{zheng2018dags} and is implemented in R package {\tt gnlearn}. However, it worked very poorly and is hence omitted.
Since the data is temporal, we also tried to decorrelate the data first marginally and then apply algorithms like NO-TEARS, PC, or LINGAM. 
However, the performance worsened or remained relatively the same compared to no adjustment. It may be because the overall covariance of the data is not separable, as different latent time series possess different covariance kernels.
Hence, the presented results are based on the direct application of the alternative methods to the simulated data. 

We also fit the NOTEARS method with adaptive LASSO penalty and employ Algorithm~\ref{algo1}, replacing $\bW^{(t)}\bY$ with $\bY$ in Step 3. Thresholding similar to step 4 is also applied to estimate the DAG. We call this A-NOTEARS, and it works better than the original NOTEARS with LASSO penalty. The rank-PC is fitted following \citep{harris2013pc}.
We compare different methods both in terms of estimation MSE for $\bW$ and Mathew's Correlation Coefficient (MCC) in identifying the true DAG.

In Section~\ref{sec:singlesimu}, we compare the method with some of the existing causal relation identification methods for time series data such as DYNOTEAR \citep{pamfil2020dynotears}, PCMCI \citep{runge2019detecting}, and PCMCIplus \citep{runge2020discovering}. 
The first method is fitted applying the Python package {\tt causalnex} and the last two are fitted applying the Python package {\tt Tigramite}. Since the software implementations for these alternative methods are designed for $S=1$ case, we generate such data for this comparison.

\subsection{Simulation setting 1}

The latent univariate stationary series are generated from Gaussian Processes with an exponential kernel. These Gaussian processes only differ in the range parameter, which is uniformly generated from (0, 10).

\begin{table}[htbp]
\centering
\caption{Estimation MSE in estimating $\bW$ when $\bZ_{\ell}$'s are generated following exponential covariance.}
\resizebox{0.9\textwidth}{!}{
\begin{tabular}{|c|ccc|ccc|}
\hline
    Time points & \multicolumn{3}{c|}{Expected neighbors = 2} & \multicolumn{3}{c|}{Expected neighbors = 4} \\
    \hline
    & DAG-OUT & A-NOTEARS & LINGAM & DAG-OUT & A-NOTEARS & LINGAM \\ 
    \hline
32 & 0.004 & 0.02 & 0.02 & 0.04 & 0.08 & 0.09  \\ 
48 & 0.003 & 0.02 & 0.01 & 0.04 & 0.07 & 0.07 \\ 
\hline
\end{tabular}
}
\end{table}

\begin{table}[htbp]
\centering
\caption{MCC when $\bz_{\ell}$'s are generated following Gaussian process with exponential kernel.}
\begin{tabular}{|c|ccccc|}
\hline
    Time points &\multicolumn{5}{|c|}{Expected neighbors = 2} \\
    \hline
  \hline
 & DAG-OUT & A-NOTEARS & LINGAM & PC & rank-PC \\ 
  \hline
32 & 0.91 & 0.59 & 0.62 & 0.36 & 0.29   \\ 
  48 & 0.92 & 0.61 & 0.59 & 0.41 & 0.39\\ 
   \hline
   & \multicolumn{5}{|c|}{Expected neighbors = 4}\\
   \hline
    & DAG-OUT & A-NOTEARS & LINGAM & PC & rank-PC  \\ 
  \hline
32 & 0.61 & 0.46 & 0.41 & 0.28 & 0.22  \\ 
  48 & 0.62 & 0.48 & 0.45 & 0.34 & 0.30  \\ 
  \hline
\end{tabular}
\end{table}

\subsection{Simulation setting 2}
In this section, the exponential covariance is replaced by the covariance kernel $$K(t_1,t_2)=\sum_{h=1}^{M}a_h\cos(h\pi|t_1-t_2|),$$ 
where $t_1,t_2\in [0,1]$ and set $a_h=1/h^2$. For different univariate time series, $M$ is sampled from $\{1,\ldots,T\}$ uniformly at random. Different values of $M$ induce different degrees of smoothness.

\begin{table}[htbp]
\centering
\caption{Estimation MSE in estimating $\bW$ when $\bz_{\ell}$'s are generated following cosine covariance kernel.}
\resizebox{0.9\textwidth}{!}{
\begin{tabular}{|c|ccc|ccc|}
\hline
    Time points &\multicolumn{3}{|c|}{Expected neighbors = 2} & \multicolumn{3}{|c|}{Expected neighbors = 4}\\
    \hline
  \hline
 & DAG-OUT & A-NOTEARS & LINGAM & DAG-OUT & A-NOTEARS & LINGAM \\ 
  \hline
32 & 0.02 & 0.03 & 0.02 & 0.05 & 0.10 & 0.10   \\ 
  48 & 0.01& 0.02 & 0.02 & 0.04 & 0.08 & 0.08\\ 
   \hline
\end{tabular}
}
\end{table}

\begin{table}[htbp]
\centering
\caption{MCC when $\bZ_{\ell}$'s are generated following a Gaussian process with a cosine covariance kernel.}
\begin{tabular}{|c|ccccc|}
\hline
    Time points &\multicolumn{5}{|c|}{Expected neighbors = 2} \\
    \hline
  \hline
 & DAG-OUT & A-NOTEARS & LINGAM & PC & rank-PC \\ 
  \hline
32 & 0.81 & 0.46 & 0.53 & 0.27 & 0.29    \\ 
  48 & 0.83 & 0.50 & 0.54 & 0.42 & 0.37\\ 
   \hline
   & \multicolumn{5}{|c|}{Expected neighbors = 4}\\
   \hline
    & DAG-OUT & A-NOTEARS & LINGAM & PC & rank-PC  \\ 
  \hline
32 &  0.55 & 0.41 & 0.38 & 0.24 & 0.18  \\ 
  48 & 0.63 & 0.45 & 0.41 & 0.25 & 0.32  \\ 
  \hline
\end{tabular}
\end{table}

\subsection{Simulation setting 3}
The latent univariate stationary series are generated from ARMA(1,1) model as $Z_{i,t}= \phi Z_{i,t-1}+\theta \varepsilon_{i, t-1}+\varepsilon_{i, t}$ with $\varepsilon_{i, t-1}\sim\Normal(0, \sigma_e^2)$ with $\sigma_e^2={(1-\phi^2)}/{(1+2\theta\phi+\theta^2)}$.
We generate $\theta,\phi\sim\Unif((0.9,1)\cup(-1, -0.9))$ to generate time series with strong dependence.

\begin{table}[htbp]
\centering
\caption{Estimation MSE in estimating the precision matrix of dimension $40\times 40$ when $\bZ_{\ell}$'s are generated following causal ARMA(1,1) models.}
\resizebox{0.9\textwidth}{!}{
\begin{tabular}{|c|ccc|ccc|}
\hline
    Time points &\multicolumn{3}{|c|}{Expected neighbors = 2} & \multicolumn{3}{|c|}{Expected neighbors = 4}\\
    \hline
  \hline
 & DAG-OUT & A-NOTEARS & LINGAM & DAG-OUT & A-NOTEARS & LINGAM \\ 
  \hline
32 & 0.01 & 0.02 & 0.02 & 0.05 & 0.10 & 0.13 \\ 
  48 & 0.01 & 0.02 & 0.02 & 0.03 & 0.06 & 0.12 \\ 
   \hline
\end{tabular}
}
\end{table}

\begin{table}[htbp]
\centering
\caption{MCC when $\bZ_{\ell}$'s are generated following causal ARMA(1,1) models.}
\begin{tabular}{|c|ccccc|}
\hline
    Time points &\multicolumn{5}{|c|}{Expected neighbors = 2} \\
    \hline
  \hline
 & DAG-OUT & A-NOTEARS & LINGAM & PC & rank-PC \\ 
  \hline
32 & 0.77 & 0.60 & 0.53 & 0.28 & 0.34 \\ 
  48 & 0.82 & 0.63 & 0.40 & 0.39 & 0.37 \\ 
   \hline
   & \multicolumn{5}{|c|}{Expected neighbors = 4}\\
   \hline
    & DAG-OUT & A-NOTEARS & LINGAM & PC & rank-PC  \\ 
  \hline
32 & 0.61 & 0.44 & 0.32 & 0.24 & 0.25 \\ 
  48 & 0.65 & 0.48 & 0.40 & 0.28 & 0.38 \\
  \hline
\end{tabular}
\end{table}

All the simulation results favour DAG-OUT both in terms of estimation MSE and MCC. The MCC is in general lower when the expected neighbors set to a higher value due to increasing complexity. In general, with increasing sample size, the estimation MSE is lower and MCC is higher.

\subsection{Single task simulation}
\label{sec:singlesimu}
In the previous simulations, following the structure of our motivated application, our generated data were matrix-variate with $S=15$.
Here, we consider a single task case with $S=1$ to compare the proposed projection-posterior based approach with DYNOTEAR \citep{pamfil2020dynotears}, PCMCI \citep{runge2019detecting}, and PCMCIplus \citep{runge2020discovering} as the available softwares only support this setting.
Here we generate the data from all three simulation settings from the preceding subsections with different choices of expected neighbors and sample size, and compare the MCCs.
We use the default tuning parameter settings for DYNOTEAR, PCMCI, and PCMCIplus, as there is no clear direction in how to choose them. Although these alternative approaches also provide across-lag causal relations, we focus here only on the marginal or lag-0 associations.
In this simulation experiment, we only consider $T=48$.
Here, DAG-OUT is still the best in all simulation scenarios, particularly when the `expected neighbors' is set at 4. The worst performance of DYNOTEAR may be due to the default tuning parameter setting from \cite{pamfil2020dynotears}. MCC under this $S=1$ case is, in general, lower than the previous cases, with $S=15$.
It may be since there is more information in the data regarding the underlying DAG when $S=15$.

\begin{table}[htbp]
\centering
\caption{MCC comparison with DYNOTEAR, PCMCI, and PCMCIplus for three simulation settings with different generation schemes for the latent process.}
\begin{tabular}{|c|ccccc|}
\hline
\hline
 Expected neighbors    &\multicolumn{4}{|c|}{$\bZ_{\ell}$'s are generated from Exponential Kernel} \\
    \hline
  \hline
 & DAG-OUT & DYNOTEAR & PCMCI & PCMCIplus \\ 
  \hline
 2  & 0.38 & 0.05 & 0.30  & 0.33   \\ 
   \hline
  4 & 0.29 & 0.02  & 0.16 & 0.19  \\
  \hline
\hline
     &\multicolumn{4}{|c|}{$\bZ_{\ell}$'s are generated from Cosine kernel} \\
    \hline
  \hline
 & DAG-OUT & DYNOTEAR & PCMCI & PCMCIplus \\ 
  \hline
2  & 0.35 & 0.03 & 0.27 & 0.29   \\ 
   \hline
 4  & 0.29 & 0.02  & 0.20  & 0.21  \\
  \hline
\hline
   &\multicolumn{4}{|c|}{$\bZ_{\ell}$'s are generated from ARMA(1,1)} \\
    \hline
  \hline
 & DAG-OUT & DYNOTEAR & PCMCI & PCMCIplus \\ 
  \hline
  2  & 0.42 & 0.06 & 0.33 & 0.33   \\ 
   \hline
4  & 0.36 & 0.02 & 0.18 & 0.19   \\
  \hline
\end{tabular}
\end{table}

\section{QWI data analysis}
\label{sec:realdata}
We study causal associations among different age and earning groups based on three earning-related variables, listed in Table~\ref{tab:response}. The data set is the Quarterly Workforce Indicators data published by the US Census Bureau.  %[REFERENCE; VINTAGE]  
We use the EarnS variable from Year 1993 to Year 1997 to form 6 earning/salary groups (as shown in Table \ref{tab:salgrp}), following the convention that group $i$ on average earns more than group $i-1$. Then we fit our proposed model to the data from 1998 to 2018. As a pre-processing step, we remove a linear trend and then mean-center and normalize. Here $S=3$ due to three responses. We take Age $\times$ Earning groupings, forming $8\times 6=48$ combinations, which is our $p$. The estimated directed acyclic graph (DAG) provides causal associations among these 48 combinations of ages and salaries.
We further study the resulting collapsed DAGs among the salary and age groups separately.
%Estimation of the DAG architecture in job-to-job switch data.

%We need to prove  (Causal Identifiability) for this model.

% Features of this dataset: There are several quarterly collected time-series data for different job-related variables at (State, Age group, Industry-type) [$50\times 8\times 20$, there are 50 states, 8 age-groups and 20 types of industries]. For each state, the lengths of the time-series are different. That's also another interesting feature.\\

% Possible analysis: We can pick one of the several job-related variables for which estimating DAG relation among those 20 industries make sense. Then, we can propose a model that will give (State, Age-group)-specific DAG. There may be information sharing that given an Age-group, across-state DAGs should be close or we can other features too on these (State, Age-group)-specific DAGs too.\\

% We can assume some parent ordering, relying on the output of some standard causal-discovery analysis software on the accumulated data, or propose a model that would detect the causal ordering as well.\\

% {\url{https://lehd.ces.census.gov/data/schema/j2j_latest/lehd_public_use_schema.html}}

%{\bf Tasks are three hiring-related and three earning variables.} The variables are in Table~\ref{tab:response}. 

\begin{table}[htbp]
\caption{Variables considered as multi-task}
\centering
{\resizebox{.9\textwidth}{!}{\begin{tabular}{rll}
  \hline
 & Variable & Explanation \\ 
  \hline
% 1 & EmpS & Estimate of stable jobs - the number of jobs that are held on both the first and last day of the quarter with the same employer \\ 
%   2 & EmpTotal & Estimated count of people employed in a firm at any time during the quarter \\ 
%   3 & EmpSpv & Estimate of stable jobs in the quarter before the reference quarter \\ 
  1 & EarnS & Average monthly earnings of employees with stable jobs \\ 
  2 & EarnHirAS & Average monthly earnings for workers who started a job that turned into a job lasting a full quarter \\ 
  3 & EarnHirNS & Average monthly earnings of newly stable employees \\ 
   \hline
\end{tabular}}}
\label{tab:response}
\end{table}

\begin{table}[htbp]
\centering
\caption{Industries with their corresponding salary groups. Higher salary groups earn more}
\begin{tabular}{rl}
  \hline
Salary group & Industry names \\ 
  \hline
1 & Agriculture, Forestry, Fishing and Hunting \\ 
  6 & Mining, Quarrying, and Oil and Gas Extraction \\ 
  6 & Utilities \\ 
  4 & Construction \\ 
  4 & Manufacturing \\ 
  4 & Wholesale Trade \\ 
  1 & Retail Trade \\ 
  3 & Transportation and Warehousing \\ 
  5 & Information \\ 
  6 & Finance and Insurance \\ 
  3 & Real Estate and Rental and Leasing \\ 
  5 & Professional, Scientific, and Technical Services \\ 
  5 & Management of Companies and Enterprises \\ 
  2 & Administrative and Support and Waste Management and Remediation Services \\ 
  2 & Educational Services \\ 
  3 & Health Care and Social Assistance \\ 
  2 & Arts, Entertainment, and Recreation \\ 
  1 & Accommodation and Food Services \\ 
  1 & Other Services (except Public Administration) \\ 
   \hline
\end{tabular}
\label{tab:salgrp}
\end{table}

% \begin{figure}[htbp]
% \centering
% \subfigure{\includegraphics[width = 0.6\textwidth]{Figures/graphHSEarn61.png}}
% \subfigure{\includegraphics[width = 0.15\textwidth]{Figures/legend.png}}
% %\caption{Estimated graph for the three cases}
% \caption{Estimated directed acyclic graphical connections among different variables with different salary levels and age groups.}
% \label{fig::real}
% \end{figure}

\begin{figure}[htbp]
\centering
\subfigure{\includegraphics[width = 0.6\textwidth]{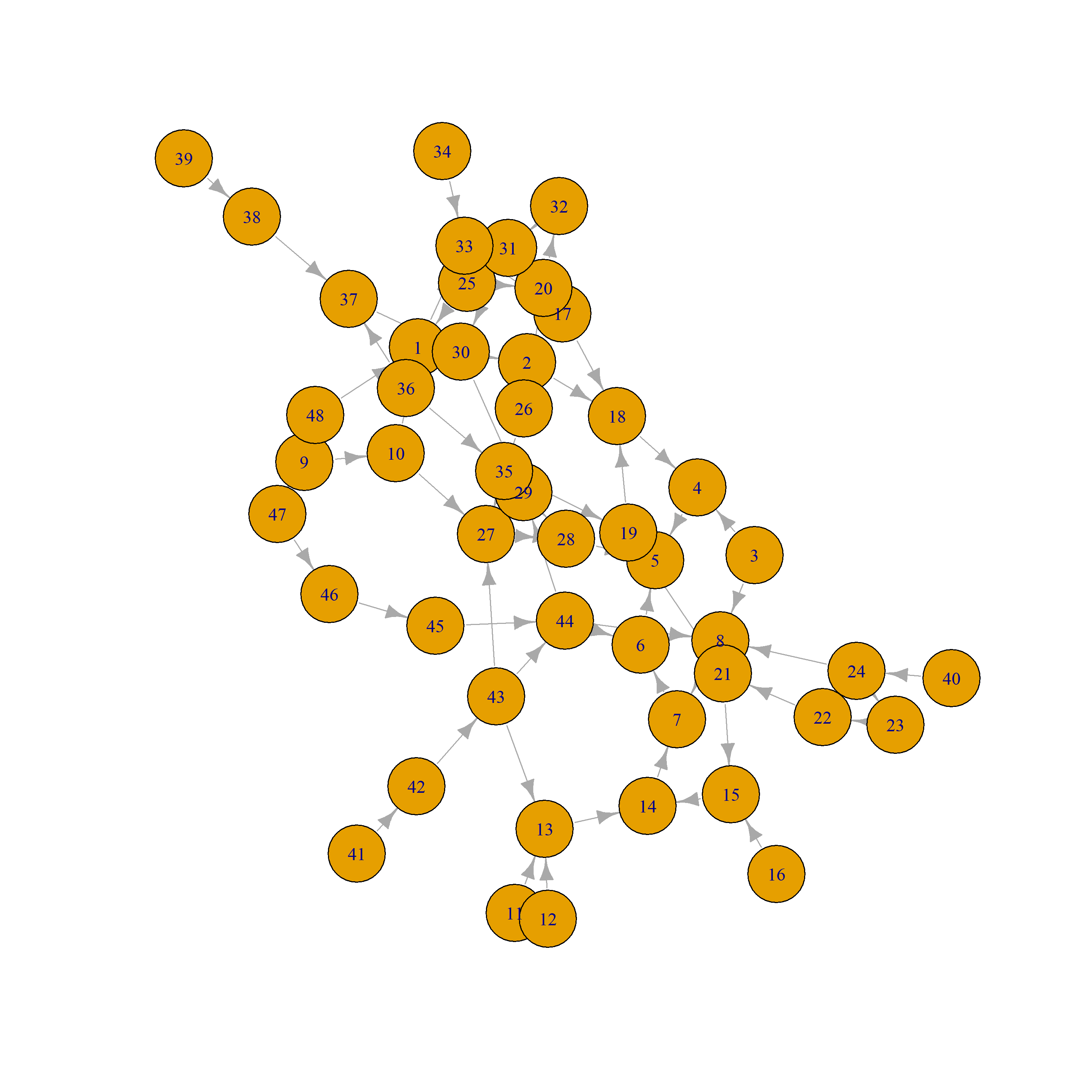}}
\subfigure{\includegraphics[width = 0.15\textwidth]{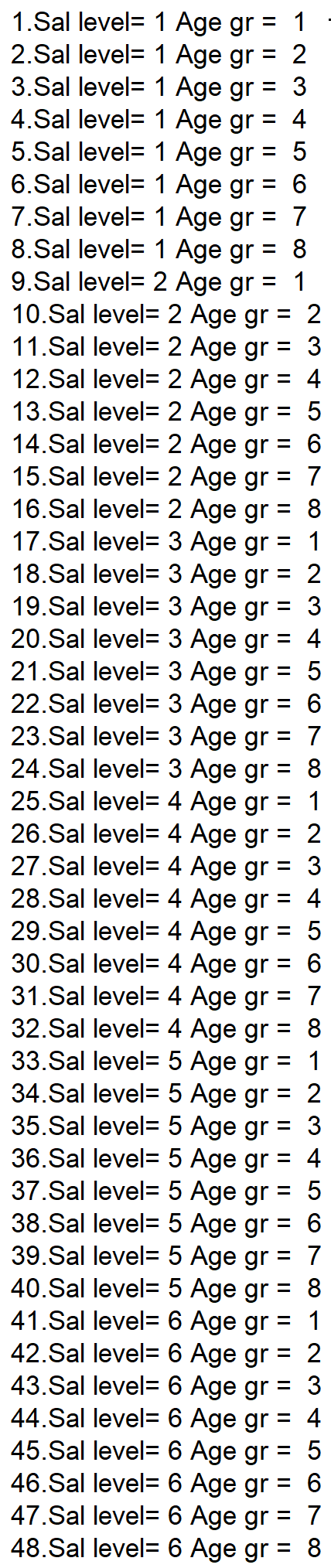}}
%\caption{Estimated graph for the three cases}
\caption{Estimated directed acyclic graphical connections among different variables with different salary levels and age groups.}
\label{fig::real}
\end{figure}

\begin{table}[htbp]
\centering
\caption{Age groups}
\begin{tabular}{ccccccccc}
  \hline
Age-groups & 1 &  2 & 3 & 4 & 5 & 6 & 7 & 8 \\ 
  \hline
Range & 14-18 & 19-21 & 22-24 & 25-34 & 35-44 & 45-54 & 55-64 & 65-99 \\ 
   \hline
\end{tabular}
\label{tab:agegrp}
\end{table}

% \begin{table}[ht]
% \centering
% \caption{Structured Hamming distance between different pairs of salary groups}
% \begin{tabular}{rrrrrrr}
%   \hline
%  & Sal level 1 & Sal level 2 & Sal level 3 & Sal level 4 & Sal level 5 & Sal level 6 \\ 
%   \hline
% Sal level 1 & 0 & 4 & 4 & 4 & 3 & 3 \\ 
%   Sal level 2 & 4 & 0 & 5 & 6 & 5 & 5 \\ 
%   Sal level 3 & 4 & 5 & 0 & 6 & 5 & 5 \\ 
%   Sal level 4 & 4 & 6 & 6 & 0 & 5 & 1 \\ 
%   Sal level 5 & 3 & 5 & 5 & 5 & 0 & 5 \\ 
%   Sal level 6 & 3 & 5 & 5 & 1 & 5 & 0 \\ 
%    \hline
% \end{tabular}
% \label{tab:salres}
% \end{table}

\begin{table}[ht]
\centering
\caption{Structured Hamming distance between different pairs of salary groups.}
\begin{tabular}{ccccccc}
  \hline
 & Sal level 1 & Sal level 2 & Sal level 3 & Sal level 4 & Sal level 5 & Sal level 6 \\ 
  \hline
Sal level 1 & 0 & 4 & 5 & 4 & 4 & 3 \\ 
  Sal level 2 & 4 & 0 & 5 & 6 & 4 & 5 \\ 
  Sal level 3 & 5 & 5 & 0 & 5 & 6 & 4 \\ 
  Sal level 4 & 4 & 6 & 5 & 0 & 5 & 1 \\ 
  Sal level 5 & 4 & 4 & 6 & 5 & 0 & 5 \\ 
  Sal level 6 & 3 & 5 & 4 & 1 & 5 & 0 \\ 
   \hline
\end{tabular}
\label{tab:salres}
\end{table}

% \begin{table}[ht]
% \centering
% \caption{Structured Hamming distance between different pairs of age groups}
% {\resizebox{.8\textwidth}{!}{\begin{tabular}{rrrrrrrrr}
%   \hline
%  & Age group 1 & Age group 2 & Age group 3 & Age group 4 & Age group 5 & Age group 6 & Age group 7 & Age group 8 \\ 
%   \hline
% Age group 1 & 0 & 3 & 3 & 2 & 2 & 2 & 2 & 2 \\ 
%   Age group 2 & 3 & 0 & 2 & 1 & 1 & 1 & 1 & 2 \\ 
%   Age group 3 & 3 & 2 & 0 & 1 & 1 & 1 & 1 & 3 \\ 
%   Age group 4 & 2 & 1 & 1 & 0 & 0 & 0 & 0 & 2 \\ 
%   Age group 5 & 2 & 1 & 1 & 0 & 0 & 0 & 0 & 2 \\ 
%   Age group 6 & 2 & 1 & 1 & 0 & 0 & 0 & 0 & 2 \\ 
%   Age group 7 & 2 & 1 & 1 & 0 & 0 & 0 & 0 & 2 \\ 
%   Age group 8 & 2 & 2 & 3 & 2 & 2 & 2 & 2 & 0 \\ 
%    \hline
% \end{tabular}}}
% \label{tab:ageres}
% \end{table}

\begin{table}[ht]
\centering
\caption{Structured Hamming distance between different pairs of age groups.}
{\resizebox{.8\textwidth}{!}{\begin{tabular}{ccccccccc}
  \hline
 & Age group 1 & Age group 2 & Age group 3 & Age group 4 & Age group 5 & Age group 6 & Age group 7 & Age group 8 \\ 
  \hline
Age group 1 & 0 & 4 & 4 & 3 & 3 & 3 & 3 & 3 \\ 
  Age group 2 & 4 & 0 & 4 & 1 & 1 & 1 & 1 & 2 \\ 
  Age group 3 & 4 & 4 & 0 & 3 & 3 & 3 & 3 & 3 \\ 
  Age group 4 & 3 & 1 & 3 & 0 & 0 & 0 & 0 & 2 \\ 
  Age group 5 & 3 & 1 & 3 & 0 & 0 & 0 & 0 & 2 \\ 
  Age group 6 & 3 & 1 & 3 & 0 & 0 & 0 & 0 & 2 \\ 
  Age group 7 & 3 & 1 & 3 & 0 & 0 & 0 & 0 & 2 \\ 
  Age group 8 & 3 & 2 & 3 & 2 & 2 & 2 & 2 & 0 \\ 
   \hline
\end{tabular}}}
\label{tab:ageres}
\end{table}

The maximum number of possible edges in the graph is 1128 resulting in a dense graph. However, in practice, the proposed model is expected to estimate a much sparser graph with a few meaningful causal connections. 
Figure~\ref{fig::real} illustrates the estimated DAG. The estimated DAG has approximately 62 connections that are deemed significant. 
Here, we see that most of the edges or causal associations are confined within a fixed salary group.
Within a given salary group, the edges would suggest casual associations across different age groups. 
Such associations are predominant.
This suggests that causal associations based on the outcomes quantifying earning numbers are primarily present within a given salary group or between two adjacent salary groups. The number of within-age edges is very limited.

For a detailed inference, Table~\ref{tab:salres} shows the structured hamming distances between different salary groups based on the summarized number of edges. 
It is defined as the minimum number of single operations, such as deletions, insertions, and re-orientations, needed to transform one DAG into another DAG. We compute these distances using the {\tt shd} function from the package {\tt pcalg} \citep{pcalgR}. The Structured Hamming distance (SHD) provides a measure of changes in the causal associations among the age groups as we move along the different salary groups.  
Table~\ref{tab:ageres}  shows the same for the age groups, measuring how the causal connections between the salary groups change from one age group to another.
%Based on Table~\ref{tab:salres},  The rest of the edges are primarily between the two adjacent salary groups.

The SHDs in Table~\ref{tab:salres} are larger than those in  Table~\ref{tab:ageres}, suggesting causal associations among different salary groups stay relatively stable across different ages. The SHDs in Table~\ref{tab:ageres} are often 0, specifically for higher age groups, meaning no differences in the estimated causal associations among the salary groups as they move from one age group to the other. Specifically, age groups 4 to 7 exhibit identical causal associations among the salary groups.

\subsection{Employment-Quantity (EarnS-EmpTotal)}
Labor economists are interested in studying the joint behavior or earnings and employment totals. 
We used the bivariate quarterly series of total employed in stable jobs and average monthly earnings from stable jobs. We extended the study of causal relations concerning earnings to the study of the causal structures of different age and salary groups concerning earnings and total employment. Figure~\ref{fig::real} shows the estimated graph of the 48 different age and salary group combinations. 
Interestingly, several path-like structures in the graph indicate a smooth monotone change of the causal structure within some age/salary groups. For example, in salary levels 5 and 6, the causal structure within the middle age groups (3-6) is a directed path. 

The SHD analysis of the nested groups of salary and age for the earnings and employment total together reveals a story similar to that for the earnings-only analysis. Notably, the causal structure among the salary groups remains stable across the middle age groups, ages 25--55. However, in the extreme age groups, below 20 years old or above 65 years old, the causal structure among the salary groups can shift significantly from that for the middle age groups.

\begin{table}[htbp]
\caption{Variables considered as multi-task}
\centering
{\resizebox{.9\textwidth}{!}{\begin{tabular}{ccl}
  \hline
 & Variable & Explanation \\ 
  \hline
1 & EmpS & Estimate of stable jobs - the number of jobs that are held on both the first and last day of the quarter with the same employer \\ 
  2 & EarnS & Average monthly earnings of employees with stable jobs \\ 
   \hline
\end{tabular}}}
\label{tab:response}
\end{table}

\begin{figure}[htbp]
\centering
\subfigure{\includegraphics[width = 0.6\textwidth]{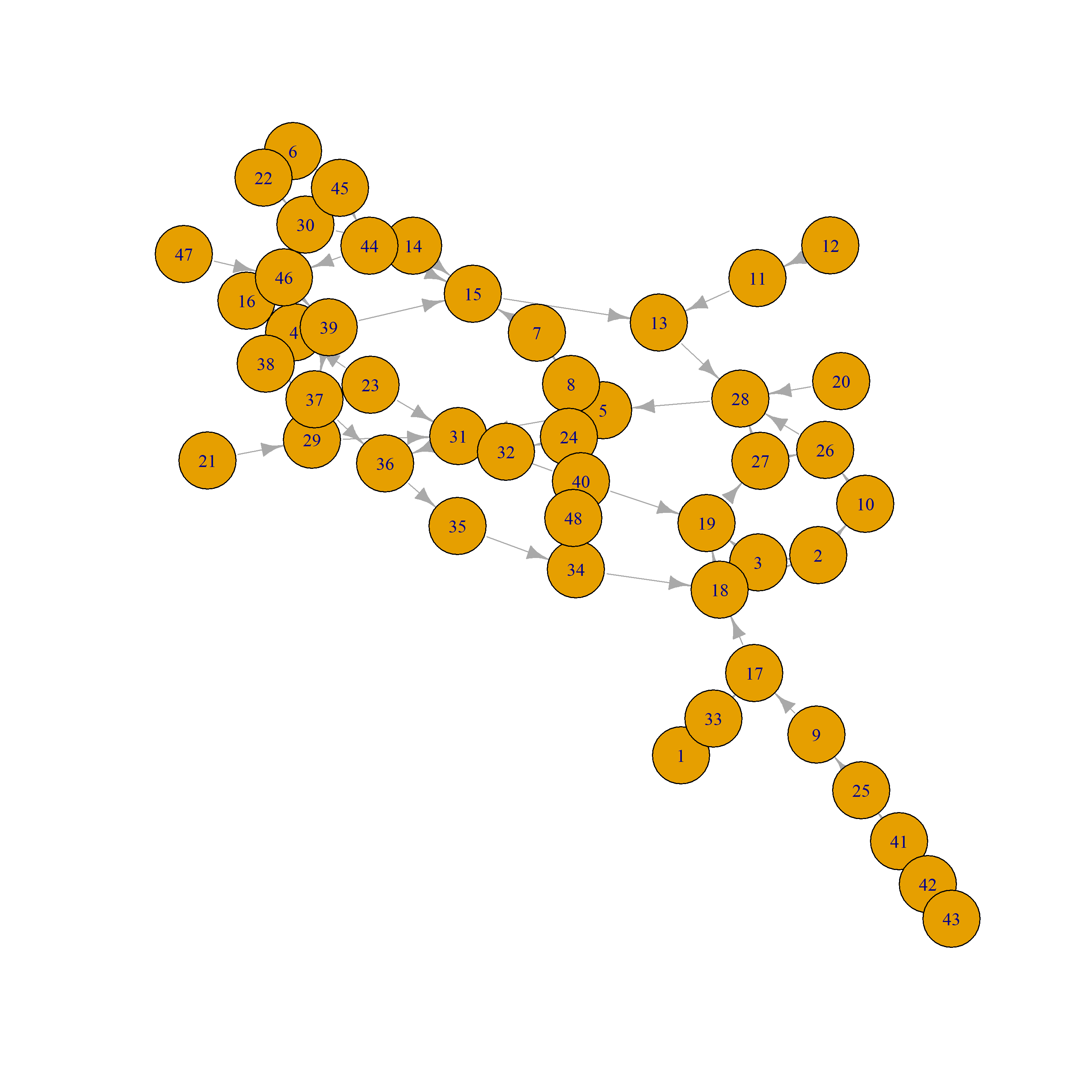}}
\subfigure{\includegraphics[width = 0.15\textwidth]{Figures/legend.png}}
%\caption{Estimated graph for the three cases}
\caption{Estimated directed acyclic graphical connections among different variables with different salary levels and age groups with Employment-Quantity.}
\label{fig::real}
\end{figure}

\begin{table}[ht]
\centering
\caption{Structured Hamming distance between different pairs of salary groups with Employment-Quantity}
\begin{tabular}{ccccccc}
  \hline
 & Sal level 1 & Sal level 2 & Sal level 3 & Sal level 4 & Sal level 5 & Sal level 6 \\ 
  \hline
Sal level 1 & 0 & 7 & 4 & 5 & 7 & 6 \\ 
  Sal level 2 & 7 & 0 & 6 & 7 & 4 & 8 \\ 
  Sal level 3 & 4 & 6 & 0 & 6 & 7 & 5 \\ 
  Sal level 4 & 5 & 7 & 6 & 0 & 7 & 8 \\ 
  Sal level 5 & 7 & 4 & 7 & 7 & 0 & 6 \\ 
  Sal level 6 & 6 & 8 & 5 & 8 & 6 & 0 \\ 
   \hline
\end{tabular}
\end{table}

\begin{table}[ht]
\centering
\caption{Structured Hamming distance between different pairs of age groups with Employment-Quantity}
{\resizebox{.8\textwidth}{!}{\begin{tabular}{ccccccccc}
  \hline
 & Age group 1 & Age group 2 & Age group 3 & Age group 4 & Age group 5 & Age group 6 & Age group 7 & Age group 8 \\ 
  \hline
Age group 1 & 0 & 6 & 7 & 6 & 7 & 6 & 8 & 8 \\ 
  Age group 2 & 6 & 0 & 5 & 5 & 6 & 6 & 5 & 6 \\ 
  Age group 3 & 7 & 5 & 0 & 1 & 2 & 2 & 3 & 2 \\ 
  Age group 4 & 6 & 5 & 1 & 0 & 1 & 3 & 2 & 3 \\ 
  Age group 5 & 7 & 6 & 2 & 1 & 0 & 4 & 3 & 4 \\ 
  Age group 6 & 6 & 6 & 2 & 3 & 4 & 0 & 5 & 4 \\ 
  Age group 7 & 8 & 5 & 3 & 2 & 3 & 5 & 0 & 5 \\ 
  Age group 8 & 8 & 6 & 2 & 3 & 4 & 4 & 5 & 0 \\ 
   \hline
\end{tabular}}}
\end{table}

\section{Discussion}
This paper proposed a projection-posterior-based Bayesian method for DAG estimation in a multivariate time-series setting. The projection, however, is generally applicable for running any Bayesian DAG estimation task.
Our approach further maintains stationarity; thus, the estimated DAG structure in the multivariate data is held at any given time.
We establish the posterior convergence properties of the proposed method and, in that, further establish two identifiability results for the unrestricted structural equation model.
A simulation and a QWI dataset illustrate the proposed method's performance. 
The estimated DAG structure allowed us to infer that earnings and total employment numbers change across different age groups and salary levels.

The proposed semiparametric time-series model for the univariate error processes can be modified to parametric models such as AR, MA, or even ARMA. This may allow us to further infer the temporal dynamics. Another future direction could be extending the posterior-projection approach to covariate-dependent DAG estimation.
Such models will be helpful for inferring covariate dependence in the DAG structure.
The temporal dynamics in the error processes can even be fully nonparametric, which is useful in several application domains such as neuroimaging, financial markets, etc. 

\section*{Appendix}
\label{sec:appendix}

 \begin{lem}[Identifiability with known $\bD$]
   \label{lem:uniqueness}
    If $(\bW_0,\bD_0, f_{0,j}: j=1,\ldots,p)$ and $(\bW_1,\bD_0, f_{1,j}: j=1,\ldots,p)$ generate the same distribution on $\bY_t$, with $\bW_0$ strictly lower triangular and $\bW_1$ having all diagonal entries zero, then $\bW_0 = \bW_1$.   
\end{lem}

\begin{proof}
    Let $\tilde \bY_t$ stand for the Fourier transform of $\bY_t$ at the $t$-th Fourier frequency. Under the regime $(\bW_0,\bD_0, f_{0,j}: j=1,\ldots,p)$, the component series of $(\bI-\bW_0)\tilde{\bY}_t$ has covariance matrix $\bD_0$ for any $t$. 
    %up to a multiple of $1/T$, entrywise. 
    Further, $(\bI-\bW_{0})\trans\bD_0^{-1/2}$ is the Cholesky factor for the precision matrix $\bOmega_0$ of $\bY_t$. Thus, by the uniqueness of the Cholesky factor, $\bD_0$ is uniquely identified from $\bOmega_0$. 

Now assume that $\bD_0=\bD=\Diag(d_1,\ldots,d_p)$, say. 
    From the equality of the precision matrices of the series under the two regimes $(\bW_0,\bD, f_{0,j}: j=1,\ldots,p)$ and $(\bW_1,\bD, f_{1,j}: j=1,\ldots,p)$ that generate the same distribution, we have $(\bI-\bW_{1})\bD^{-1}(\bI-\bW_{1})\trans=(\bI-\bW_{0})\bD^{-1}(\bI-\bW_{0})\trans$, the matrix $\bW_{1}$ must satisfy the relation $(\bI-\bW_{1})\bD^{-1/2}=(\bI-\bW_{0})\bD^{-1/2}\bP$ for some orthogonal matrix $\bP$. We then recursively show that  $\bP$ must be the identity matrix.

    Observe that $\bD^{-1/2}-\bW_{1}\bD^{-1/2}=\bD^{-1/2}\bP-\bW_{0}\bD^{-1/2}\bP$.
    By the construction, diagonal entries of $\bW_{1}\bD^{-1/2}$ are all zero and $\bW_{0}\bD^{-1/2}$ is strictly lower triangular. Thus, the first diagonal entry of $\bW_{0}\bD^{-1/2}\bP$ is zero.
    Then, comparing the first diagonal entry on both sides, we have $1/\sqrt{d_{1}} =  P_{1,1}/\sqrt{d_{1}}$ so that $P_{1,1}=1$. Since the first row $(P_{1,i}: i=1,\ldots,p)$ and the first column $(P_{i,1}: i=1,\ldots,p)\trans$ both must have unit norm by the orthogonality of $\bP$, we have     
     $P_{1,i}=P_{i,1}=0$ for all $i=2,\ldots,p$. 
     
     We repeat the process with the second diagonal entry under and obtain $P_{2,2}=1$, $P_{2,i}=P_{i,2}=0$ for all $i=1,3,4,\ldots,p$. Continuing the arguments for every diagonal entry, we obtain $\bP=\bI$. 
\end{proof}

Next, we extend the above argument to show that the model parameters are continuously identifiable if $\bD$ is known and equal to a given value $\bD_0$. 

 \begin{lem}
   \label{lem:soundness}
    Let $(\bW_0,\bD_0, f_{0,j}: j=1,\ldots,p)$ and $(\bW,\bD_0, f_{j}: j=1,\ldots,p)$ be two sets of parameters with the same diagonal matrix, where $\bW_0$ is strictly lower triangular and $\bW$ has all diagonal entries zero. If the  
    precision matrices for process $\bY_t$ the generate are $\bOmega_0$ and $\bOmega$ respectively, then 
    \begin{align}
        \|\bW-\bW_0\|_\mathrm{F} \lesssim \sqrt{p}\|\bOmega-\bOmega_0\|_\mathrm{F}^{1/2}, 
    \end{align}
    where $\bOmega^{1/2}$ and $\bOmega_0^{1/2}$ are positive definite square roots of $\bOmega$ and $\bOmega_0$ respectively, and the implied constant of proportionality in `$\lesssim$' may depend on $(\bW_0,\bD_0, f_{0,j}: j=1,\ldots,p)$. 
\end{lem}

\begin{proof}
Let $\eta=\|\bOmega^{1/2}-\bOmega_0^{1/2}\|_\infty$. We may write $\bI-\bW =\bD_0^{1/2} \bOmega^{1/2}\bP$ and $\bI-\bW_0=\bD_0^{1/2}\bOmega_0^{1/2}\bP_0$ for some orthogonal matrices $\bP$ and $\bP_0$. Then 
\begin{align}
\|\bW-\bW_0\|_{\mathrm{F}}\lesssim &\Frob{\bOmega^{1/2}\bP-\bOmega_0^{1/2}\bP_0} \nonumber \\ 
&\le \op{\bP} \Frob{\bOmega^{1/2}-\bOmega_0^{1/2}}+\op{\bOmega_0}^{1/2} \Frob{\bP-\bP_0}\nonumber\\
&\lesssim \Frob{\bOmega-\bOmega_0}+ \Frob{\bP-\bP_0},
\end{align}
using the facts that $\op{\bP}=1$ by its orthogonality and $\Frob{\bOmega^{1/2}-\bOmega_0^{1/2}}$ and $\Frob{\bOmega-\bOmega_0}$ are of the same order. 
The last assertion holds because $\bOmega_0$ has all eigenvalues $\lambda_{0,1},\ldots,\lambda_{0,p}$ bounded between two fixed positive numbers, so for any $\bOmega$ with $\Frob{\bOmega-\bOmega_0}$ sufficiently small, the corresponding eigenvalues $\lambda_{1},\ldots,\lambda_{p}$ are also bounded between two fixed positive numbers and 
$$\Frob{\bOmega^{1/2}-\bOmega_0^{1/2}}^2=\sum_{j=1}^p (\sqrt{\lambda_i}-\sqrt{\lambda_{0,i}})^2 \asymp 
\sum_{j=1}^p (\lambda_i-\lambda_{0,i})^2=\Frob{\bOmega-\bOmega_0}^2.$$ 
Thus, it remains to show that $\Frob{\bP-\bP_0} \lesssim \sqrt{p} \Frob{\bOmega-\bOmega_0}^{1/2}$. 

Write $\tilde{\bP}=\bP_0\trans\bP=(\!(\tilde{P}_{ij})\!)$. Then 
$(\bI-\bW_0)\trans \bD_0^{-1/2}\tilde{\bP}=\bOmega_0^{1/2}\bP$ while $(\bI-\bW)\trans \bD_0^{-1/2}=\bOmega^{1/2}\bP$. Therefore, it follows that  
\begin{align}
\label{transfer identity}
\bD_0^{-1/2}-\bW \bD_0^{-1/2}=\bD_0^{-1/2}\tilde{\bP}-\bW_0 \bD_0^{-1/2}\tilde{\bP}+\bDelta,
\end{align}
where $\Frob{\bDelta}=\Frob{\bOmega^{1/2}-\bOmega^{1/2}}=\eta$, say. Since the first diagonal entry of $\bW$ is zero by definition, by comparing the $(1,1)$th entry of both sides, it follows from \eqref{transfer identity} that $|d_1^{-1/2} \tilde{P}_{11} -d_1^{-1/2}|\le \eta$. We can then get the final bound in terms of Frobeneous distance on $\bDelta$ directly; here we used the fact that the entire first row of $\bW_0$ consists of zero. Thus $|1-\tilde{P}_{11}|\le C\eta$ for some constant $C>0$. Now, as the rows and columns of $\tilde{P}$ have unit norm, it also follows that $\sum_{j\ne 1} \tilde{P}_{1j}^2\lesssim \eta$ and $\sum_{j\ne 1} \tilde{P}_{j1}^2\lesssim \eta$. 

Next, we look at both sides' $(2,2)$th entry. Owing to the fact the entries of $\bW_0$ are bounded in magnitude, and $|\tilde{P}_{12}|\lesssim \sqrt{\eta}$, it follows that 
$ |d_1^{-1/2} W_{0,21}\tilde{P}_{12} + d_2^{-1/2} \cdot 1\cdot \tilde{P}_{22}-d_2^{-1/2} |\le \eta $
implying that $|1-\tilde{P}_{22}|\le \eta +C\sqrt{\eta}$, and using the unit norm property of the second row and second column, it follows that $\sum_{j\ne 2} \tilde{P}_{2j}^2\lesssim \eta$ and $\sum_{j\ne 2} \tilde{P}_{j2}^2\lesssim \eta$. 

Continuing this way and noting that at most $p$ entries can be there in any row, using the Cauchy-Schwarz inequality, it follows for all $i=1,\ldots,p$, $|1-\tilde{P}_{i,i}|\le \eta+C\sqrt{\eta}$ and $\sum_{j\ne i} \tilde{P}_{ij}^2\lesssim \eta$. Therefore, $\Frob{\bI-\tilde{\bP}}^2\lesssim p(\eta+C\sqrt{\eta})^2 +p\eta$. Hence $\Frob{\bP-\bP_0}\lesssim \sqrt{p\eta}$, establishing the assertion.

\end{proof}

%{\color{red} Some exploration} The precision at $\omega$ is $(\bI-\bW)\bD^{-1}\bS^{-1}(\omega)\bD^{-1}(\bI-\bW)\trans=(\bI-\bW_0)\bD_0^{-1}\bS_0^{-1}\bD_0^{-1}(\omega)(\bI-\bW_0)\trans$. Then comparing determinants we have $\prod_{j}f_{j}(\omega)=c\prod_{j}f_{0,j}(\omega)$, where $c=\det(\bI-\bW)^2\det(\bD_0)/\det(\bD)$. Since $\int_{\omega}f_{j}(\omega)=1$, we have $c=1$. Thus spectral densities' product will always match; thus, for iid case, they will match.

\begin{lem}[Identifiability with unknown $\bD$]
\label{lem:identifiability without}
    Let $(\bW_0,\bD_0, f_{0,j}: j=1,\ldots,p)$ and $(\bW_1,\bD_1, f_{1,j}:j=1,\ldots,p)$ lead to the same distributions for the process $(\bY_t: t=1,2,\ldots,T)$ for all $T=1,2,\ldots$, where all diagonal elements of $\bW_1$ are zero and $\bW_0$ is a strictly lower triangular matrix. If $\{f_{k,j}:j=1,\ldots,p\}$ are linearly independent sets of functions for $k=0,1$, then $\bW_0=\bW_1$, $\bD_0=\bD_1$ and $f_{0,j}=f_{1,j}$ for all $j=1,\ldots,p$. 
\end{lem}

\begin{proof}
By the equality of the stationary marginal distribution of $\bY_t$, we have $(\bI-\bW_1)\bD_1(\bI-\bW_1)\trans=(\bI-\bW_0)\bD_0(\bI-\bW_0)\trans$. 

Let $\tilde{\bY}_t$ be the Fourier transformation of the original time series at the $t$th frequency. 
    Under the regime $(\bW_0,\bD_0, f_{0,j}: j=1,\ldots,p)$, the components of the time series $(\bI-\bW_0)\tilde{\bY}_t$ are independent. By the assumption, the covariance matrix of 
    $(\bI-\bW_1)\tilde{\bY}_t=(\bI-\bW_1)(\bI-\bW_0)^{-1}(\bI-\bW_0)\tilde{\bY}_t$ is diagonal. 

    Let $\bA=(\bI-\bW_1)(\bI-\bW_0)^{-1}=(\!(a_{j,k})\!)$, which is non-sigular. Then  $\bA\bD_1\bS_{1}\bD_{1}\bA\trans=\bD_{0}\bS_{0}\bD_{0}$ is a diagonal matrix, where $\bS_1$ and $\bS_0$ are the diagonal covariance matrices of the vector of Fourier transforms at all Fourier frequencies corresponding to $T$ temporal observations, as described by \eqref{eq:st}, respectively under the distributional regimes of $(\bW_1,\bD_1, f_{1,j}: j=1,\ldots,p)$ and $(\bW_0,\bD_0, f_{0,j}: j=1,\ldots,p)$. Hence, for all $T=1,2,\ldots$, and all Fourier frequencies $\omega_{m,T}$, $m=1,\ldots,T$, 
    $$\sum_{j=1}^p d_{1,j}f_{1,j}(\omega_{m,T}) a_{j,\ell}a_{j,k} = 0\mbox{ for all }\ell\neq k.$$ 
    Since the set of all Fourier frequencies corresponding to all $T$ is dense and spectral densities are continuous, $\sum_{j=1}^p d_{1,j}f_{1,j}(\omega_{m,T}) a_{j,\ell}a_{j,k} = 0$ for all $\ell\neq k$ and all $\omega_{m,T}$. The assumed linear independence of the spectral density functions now implies that $a_{j,\ell}a_{j,k}=0$ for all  $\ell\neq k$. Comparing the diagonal entries, we obtain $\sum_{j=1}^p d^2_{1,j}f_{1,j}(\omega_{m,T})a^2_{j,\ell}=d^2_{0,\ell}f_{0,j}(\omega_{m,T}) \neq 0$ for all $\ell=1,\ldots,p$. Hence for any row $j=1,\ldots,p$, exactly one entry $a_{j,\ell}$ is non-zero, say at $\ell=\ell_j$. We note that $\ell_j$'s are different for different $j$, since otherwise at least one column of $\bA$ will be entirely zero, rendering $\bA$ to be singular. Hence $\bA$ is a diagonal matrix up to a row-permutation. 
    
Let $\bP$ be a permutation matrix such that $\bP(\bI-\bW_1)(\bI-\bW_0)^{-1}=\bD$ 
is a diagonal matrix. Then $\bP(\bI-\bW_1)=\bD(\bI-\bW_0)$ is lower triangular. Since the diagonal entries of $\bW_1$ are all zero, we must have $\bP=\bI$. Hence $\bA=\bD$ is diagonal and $(\bI-\bW_1)=\bD(\bI-\bW_0)$. A comparison of diagonal entries now gives $\bD=\bI$ since $\bW_1$ and $\bW_0$ have zero diagonal entries. Thus $\bW_1=\bW_0$, and hence $\bD_1=\bD_0$. Therefore, the latent process $(\bZ_t: t=1,\ldots,T)$ has the same distribution for all $T$ under both distributional regimes, implying that they have the sets of spectral densities  $(f_{1,j}: j=1,\ldots,p)$ and $(f_{0,j}: j=1,\ldots,p)$ are identical. 
\end{proof}

%{\color{red} (Even under our low-rank parametrization, the spectral densities are still linear-independent due to element-wise non-linearity.) }
The following result gives a characterization of the condition of linear independence of the spectral densities of the components of the latent process in terms of the linear independence of the spectral densities of the components of the observed process, which is empirically testable. Let $\odot$ refer to the Hadamard (entry-wise) product of matrices of equal sizes. 

\begin{lem}
\label{lem:characterize}
Consider the distributional regime $(\bW,\bD, f_j: j=1,\ldots,p)$ for the observed time series $(\bY_t: t=1,\ldots,T)$, $T=1,2,\ldots$. 
    If $(\bI-\bW)^{-1}\odot(\bI-\bW)^{-1}$ is invertible, then the  component-wise marginal spectral densities of $(\bY_{t}:t=1,2,\ldots)$ are linearly independent if and only if $\{f_{1},\ldots,f_{p}\}$ are linearly independent.
\end{lem}

\begin{proof}
Let $\bGamma(h)=\mathrm{E}(\bY_t \bY_{t+h}\trans)$ be the autocovariance function of the observed process $(\bY_t: t=1,2,\ldots)$ and $\bm{\Upsilon}(h)=\mathrm{E}(\bZ_t \bZ_{t+h}\trans)$ be the autocovariance function of the latent process $(\bZ_t: t=1,2,\ldots)$. Using the representation \eqref{eq:SEM}, it follows that $\bGamma(h)= \bA\bD^{1/2}\bm{\Upsilon}(h) \bD^{1/2}\bA\trans$, where $\bA=(\!(a_{j,k})\!)=(\bI-\bW)^{-1}$. Taking the Fourier transform on both sides, for all frequencies $\omega$, the joint spectral density matrix $\bLambda(\omega)$ of $(\bY_t: t=1,2,\ldots)$ is given by 
\begin{align*}
\bLambda(\omega)=\bA \bD \Diag(f_1(\omega),\ldots,f_p(\omega)) \bA\trans. 
\end{align*}
Hence, the marginal spectral density of $Y_j(t)$ is given by $\lambda_j(\omega)=\sum_{k=1} d_k a_{j,k}^2 f_k(\omega)$, $j=1,\ldots,p$, that is, 
$$(\lambda_1(\omega),\ldots,\lambda_p(\omega))\trans=(\bA\odot \bA) \bD(f_1(\omega),\ldots,f_p(\omega))\trans.$$ 
Since $\bA\odot \bA$ is assumed to be non-singular and $\bD$ has all entries positive, it follows that the set of functions $ \{ \lambda_1(\omega),\ldots,\lambda_p(\omega)\}$ is linearly independent if and only if $\{f_1(\omega),\ldots,f_p(\omega)\}$ is so. 
\end{proof}

\begin{proof}[Proof of Theorem~\ref{contraction rate}]
We prove the result in two parts. In the first part, using the general theory of posterior contraction, we show that the unrestricted posterior contracts at the rate $\epsilon_T$ described above. In the second part, using a concentration inequality of a projection posterior, we show that the contraction rate is inherited from the unrestricted posterior to the projection posterior. 

We can rewrite our model as $\bY_t=(\bI-\bW)^{-1}\bD\bZ_t$ which is similar to the OUT model in \cite{roy2024bayesian}, except that there is no additional orthogonal transformation applied to $\bZ_t$, that is, the orthogonal transformation in the OUT model is fixed at the identity. Also, in the unrestricted model for the present problem, unlike \cite{roy2024bayesian}, the matrix $\bW$ is not strictly lower triangular; only the true one is after a simultaneous row-column permutation. However, this does not affect the arguments used in the proof. 

The precision matrix of $\bY_t$ is given by $\bOmega=(\bI-\bW)\trans\bD^{-1}(\bI-\bW)$ and its true value $\bOmega_0=(\bI-\bW_0)\trans\bD_0^{-1}(\bI-\bW_0)$. First, we establish the posterior contraction rate $\epsilon_T$ for $\bOmega$ at $\bOmega_0$: 
\begin{align} 
\label{rate omega}
\Pi(\|(\bI-\bW)\trans\bD^{-1}(\bI-\bW)-(\bI-\bW_0)\trans\bD_0^{-1}(\bI-\bW_0)\|_2>m_T \epsilon_T\mid \bY)\rightarrow 0 
\end{align}
for every $m_T\to\infty$. 

To this end, we need to show that $\epsilon_T$ is the pre-rate on prior concentration given by Condition (i) of the general pseudo-posterior contraction rate result given by Theorem~1 of \cite{roy2024bayesian}. This variation of the standard rate result Theorem~8.19 of \cite{ghosal2017fundamentals} for non-i.i.d. observation is needed because we have used the Whittle likelihood to define the posterior distribution instead of the actual (far more complicated) likelihood. To this end, we proceed as in the proof of Theorem~1 of \cite{roy2024bayesian} and bound the Frobenius distance for $\bOmega$ in terms of $\bD$ and $\bW$. Since the true $\bW_0$ has only $s$ non-zero entries, the horseshoe prior has log-prior concentration rate $s\log (1/\epsilon)$ by standard estimates like Lemma A.4 from \cite{song2023nearly} or Lemma 6 of \cite{bernardo1998information} under the assumed conditions on the hyperparameters. For $\bD$, as we use the clustering encouraging prior through the stick-breaking process, the estimate for i.i.d. prior used in \cite{roy2024bayesian} cannot be directly used. Nevertheless, as clustering only reduces dimension, the induced prior on the number of clusters $K$ from the Pitman-Yor stick-breaking process decays only like $e^{-c K\log p}$ and $\log p\lesssim \log T\lesssim \log(1/\epsilon_T)$, the estimate of log-prior concentration rate $p\log (1/\epsilon)$ from \cite{roy2024bayesian} remains valid. The prior concentration estimate for the spectral densities is directly used from \cite{roy2024bayesian}. Combining these, the pre-rate $\epsilon_T$ is established. 
Hence, the posterior concentration rate for $\bOmega$ at $\bOmega_0$ in terms of the Frobenius distance is $\epsilon_T$. 

Next, we translate this conclusion to a result on posterior contraction for $\bW$ at $\bW_0$. 

If $\bD$ is known as $\bD_0$, then by Lemma~\ref{lem:soundness}, the posterior for $\bW$ contracts at $\bW_0$ at the rate $\sqrt{p\epsilon_T}$. 

Now consider the scenario that $\bD$ is not known. As the posterior for $\bOmega$ contracts at $\bOmega_0$ and the map $\bOmega\mapsto (\bW, \bD, f_j: j=1,\ldots,p)$ is identifiable by Lemma~\ref{lem:identifiability without}, it remains to show the continuity of the map.

% \begin{lem}
%     If a continuous function $h$ is injective on an open set $U$, then $h^{-1}$ is continuous on $h(U)$.
% \end{lem}

For a matrix $\bA$, let $\diag(\bA)$ stand for the vector of its diagonal elements.  
Consider the forward map $h$  taking the value 
$$\mathrm{LT}((\bI-\bW)\trans\bD^{-1}(\bI-\bW))=(\!(\be_{i}\trans(\bI-\bW)\trans\bD^{-1}(\bI-\bW)\be_j): i\le j\!)$$ 
as a function of $\bw=\vect^{*}(\bW)=\vect(\bW)\setminus \diag(\bW)$ 
since the diagonals are fixed at zero, and $\bd=\diag(\bD)$, 
where $\mathrm{LT}$ stands for the operation of extracting the lower triangular entries of a symmetric matrix (thus avoiding duplication) and $\be_{i}$ is the unit vector with one at the $i$th position. For a given sufficiently small $\epsilon>0$, consider $V=\{\bOmega:\|\bOmega-h(\bw_0,\bd_0)\|_2<\epsilon\}$, where $\bw_0$ and $\bd_0$ are the true values of $\bw$ and $\bd$. It suffices to note that the map is differentiable with a non-zero gradient at $(\bw_0,\bd_0)$, so that the inverse map will be one-to-one and differentiable, and hence also continuous, on $V$. Observe that the derivatives are given by 
\begin{align*}
   &\frac{\partial \{\be_{i}\trans(\bI-\bW)\trans\bD^{-1}(\bI-\bW)\be_j\}}{\partial \bw}=-\mathrm{LT}(\{\bD^{-1}(\bI-\bW)+(\bI-\bW)\trans\bD^{-1}\}\be_{j}\be_{i}\trans)\\ 
   &\frac{\partial\, \mathrm{LT}((\bI-\bW)\trans\bD^{-1}(\bI-\bW))}{\partial d_{i,i}} = -d_{i,i}^{-2} \vect(\bh_{i}\bh_{i}\trans)\trans,
\end{align*}
where $\bH=(\bI-\bW)$ and we have $\bH\trans\bD^{-1}\bH = \sum_{i=1}^{p} d_{i,i}^{-1}\bh_{i}\bh_{i}\trans$. 

The gradient vector can be arranged as a  $p^2\times p(p+1)/2$-matrix $\bJ=(\bJ_{1}\trans,\bJ_2\trans)\trans$given by  
\begin{align*} 
\bJ_1 & = (\!(-\mathrm{LT}(\{\bD^{-1}(\bI-\bW)+(\bI-\bW)\trans\bD^{-1}\}\be_{j}\be_{i}\trans): {i\le j})\!), \\
\bJ_2 & =-(d_{1,1}^{-2} \,\mathrm{LT}(\bh_{1}\bh_{1}\trans),\ldots,d_{p,p}^{-2} \,\mathrm{LT}(\bh_{p}\bh_{p})\trans)\trans.
\end{align*} 
Note that, the point $(\bW_{0},\bd_0)$ belongs to $\mathbb{R}^{p(p+1)/2}$ due to the DAG-ness condition. Hence, by the general Implicit Function Theorem (cf., Theorem~3.1 of  \cite{blackadar2015general}) applied to 
$$H(\bw,\bd, \bg)=h(\bw,\bd)-\bg=\mathrm{LT}((\bI-\bW)\trans\bD^{-1}(\bI-\bW))-\bg,$$ 
the map $h^{-1}$ continuous around a small neighborhood $\N$ of $(\bw_{0},\bd_{0},\bg_{0})$ such that whenever $(\bw,\bd,\bg)\in\N$ we have $H(\bw,\bd,\bg)=0$, where $\bg_0=\mathrm{LT}((\bI-\bW_{0})\trans\bD_{0}^{-1}(\bI-\bW_{0}))$.

It remains to show that the posterior induced by the map $\bW\mapsto \bW^*$ given by \eqref{projection operator} inherits the concentration properties of the unrestricted posterior for $\bW$. 
Here $\bW^*=\arg\min_{\bbeta} L(\bW,\bbeta) + P_{\lambda,\alpha,\rho}(\bbeta)$, where $P_{\lambda,\alpha,\rho}(\bbeta)$ is the penalty part and $L(\bW,\bbeta)=\frac{1}{T}\|\bW\bY-\bbeta\bY\|_2^2$.
Let for a general loss $L(\bW_{1},\bW_{2})\le cL(\bW_{1},\bW_{3})+cL(\bW_{2},\bW_{3})$ for some constant $c$. In case of square-error loss like ours, we can take $c=2$.
We then automatically have
$L(\bW,\bW^*)+P_{\lambda,\alpha,\rho}(\bW^*)\le L(\bW,\bW_0)+P_{\lambda,\alpha,\rho}(\bW_0)$, giving that $L(\bW_0,\bW^*)\le (c+1)L(\bW,\bW_0)+cP_{\lambda,\alpha,\rho}(\bW_0)$.
For a $\bW_0$ with the true structure, such as the DAG structure, and a suitably small penalty parameter, the first term typically dominates the second term.
Note that under the truth, $\|\eE\left(\frac{1}{T}\bY\bY\trans\right) \|_{op}\le \max_i d_{0,i}$.
Hence, the posterior contraction rate is inherited by choosing the penalty terms carefully for sufficiently large $T$. Thus, the projection-posterior contraction rate for the matrix $\bW$ inducing the DAG-structure is $\epsilon_T$ characterized by \eqref{eq:contraction rate} when $\bD$ is known, and the projection posterior is consistent if $\bD$ is unknown. 
\end{proof}

   \section*{Funding}
%\begin{funding}
The authors would like to thank the National Science Foundation for the Collaborative Research Grants DMS-2210280 (for Subhashis Ghosal), DMS-2210281 (for Anindya Roy), and DMS-2210282 (for Arkaprava Roy).
%\end{funding}
  
	\bibliographystyle{plainnat}
	\bibliography{combined}

\end{document}